\documentclass[11pt]{article}
\usepackage[utf8]{inputenc}
\usepackage[margin=1in]{geometry}
\usepackage[english]{babel}
\usepackage{amsmath}
\usepackage{amsthm, amssymb}
\usepackage{amsfonts}
\usepackage{mathtools}
\usepackage{algorithm}
\usepackage{algpseudocode}
\usepackage{graphicx}
\usepackage{xcolor}     % for colour
\usepackage{mdframed}\mdfsetup{
linecolor=white,
backgroundcolor=gray!20,
}   % for framing
\usepackage{newtxtext, newtxmath}
\usepackage{bbm}
 \usepackage[
   pagebackref,
   colorlinks=true,
   urlcolor=black,
   linkcolor=black,
   citecolor=black,
 ]{hyperref}
\usepackage[nameinlink]{cleveref}
\usepackage{thmtools} 
\usepackage{thm-restate}

\usepackage[draft,inline,marginclue,index]{fixme}  %Remove draft to erase all notes, remove inline to erase inline notes.

\fxsetup{theme=color,mode=multiuser}
\FXRegisterAuthor{linda}{alinda}{\colorbox{blue}{\color{white}Linda}}  %First two arguments must be different
\FXRegisterAuthor{hb}{ahb}{\colorbox{magenta}{\color{white}Hedyeh}}

\theoremstyle{plain}
\newtheorem{theorem}{Theorem}[section]

\newtheorem{definition}[theorem]{Definition}
\newtheorem{lemma}[theorem]{Lemma}
\newtheorem{corollary}[theorem]{Corollary}
\newtheorem{proposition}[theorem]{Proposition}
\newtheorem{claim}[theorem]{Claim}

\newtheorem{example}[theorem]{Example}
\newtheorem*{note}{Note}

\newcommand{\eps}{\varepsilon}
\newcommand{\E}{\mathbb{E}}
\newcommand{\one}{\mathbf{1}}

\newcommand{\allocsub}[1]{\mathbb{A}_{#1}}
\newcommand{\allocsubi}{\allocsub{i}}

\newcommand{\inspect}[1]{\mathbb{I}_{#1}}
\newcommand{\inspecti}{\inspect{i}}

\newcommand{\U}{\mathcal{U}}%unopened
\newcommand{\M}{\mathcal{M}}%matroid %the grand set
\newcommand{\oo}{\alpha}%outside option
\newcommand{\Pro}{\mathbf{P}}%problem

\newcommand{\W}{\mathcal{W}}

\newcommand{\strike}{\sigma}
\newcommand{\covered}{\kappa}
\newcommand{\R}{\mathbb{R}}

\newcommand{\N}{\mathbb{N}}

\DeclareMathOperator*{\argmax}{arg\,max}
\DeclareMathOperator*{\argmin}{arg\,min}
\DeclarePairedDelimiter{\set}{\{}{\}}

\DeclarePairedDelimiter{\sq}{[}{]}
\DeclarePairedDelimiter{\paren}{\lparen}{\rparen}

\DeclarePairedDelimiter{\floor}{\lfloor}{\rfloor}

\newcommand{\weitz}{\mathrm{Weitz}}
\newcommand{\OPT}{\mathrm{OPT}}

\newcommand{\PNOI}{\mathbf{P}}
\newcommand{\TPNOI}{\mathbf{TP}}
\newcommand{\DTPNOI}{\mathbf{DTP}}
\newcommand{\STDP}{\mathbf{ST}}
\newcommand{\order}{\mathbf{ord}}
\newcommand{\ALG}{\mathrm{ALG}}
\newcommand{\OAL}{\mathrm{OAL}}
\newcommand{\LOAL}{\mathrm{LOAL}}
\newcommand{\poly}{\mathrm{poly}}
\newcommand{\NEG}{\mathsf{NEG}}
\newcommand{\phasetrans}[1]{\mathbb{PT}_{#1}}
\newcommand{\class}{\mathcal{C}}
\newcommand{\indexthreshold}{index-threshold sequence}
\newcommand{\GP}{\mathcal{P}}

\title{Pandora's Problem with Nonobligatory Inspection: Optimal Structure and a PTAS}
\author{Hedyeh Beyhaghi\thanks{Carnegie Mellon University. Email: \texttt{hedyeh@cmu.edu}.} \and%
Linda Cai\thanks{Princeton University. Email: \texttt{tcai@princeton.edu}. }}
\date{}

\begin{document}
\maketitle
\begin{abstract}
Weitzman~\cite{Weitzman79} introduced Pandora's box problem as a mathematical model of sequential search with inspection costs, in which a searcher is allowed to select a prize from one of $n$ alternatives. Several decades later, Doval~\cite{Doval18} introduced a close version of the problem, where the searcher does not need to incur the inspection cost of an alternative, and can select it uninspected. Unlike the original problem, the optimal solution to the nonobligatory inspection variant is proved to need adaptivity~\cite{Doval18}, and by recent work of~\cite{FuLL22}, finding the optimal solution is NP-hard.

Our first main result is a structural characterization of the optimal policy: We show there exists an optimal policy that follows only two different pre-determined orders of inspection, and transitions from one to the other at most once. Our second main result is a polynomial time approximation scheme (PTAS). Our proof involves a novel reduction to a framework developed by~\cite{FuLX18}, utilizing our optimal two-phase structure. Furthermore, we show Pandora's problem with nonobligatory inspection belongs to class NP, which by using the hardness result of~\cite{FuLL22}, settles the computational complexity class of the problem. Finally, we provide a tight 0.8 approximation and a novel proof for \emph{committing policies}~\cite{BeyhaghiK19} (informally, the set of nonadaptive policies) for general classes of distributions, which was previously shown only for discrete and finite distributions~\cite{GuhaMS08}.

\end{abstract}

\section{Introduction}\label{sec:intro}

Pandora's box problem, defined by Weitzman~\cite{Weitzman79}, is a model of sequential search, in which a searcher is presented a list of options to choose from and obtaining information about the value of each option is costly. 
More formally, in a Pandora's box problem, a searcher is allowed to select a prize from one of $n$ initially closed boxes. The values of the prizes inside the boxes are independent random variables, sampled from (not necessarily identical) distributions that are known to the searcher. The searcher chooses a sequence of operations, each of which is either opening a box or selecting a box. Opening box $i$ has an associated cost $c_i$ and results in learning the value $v_i$ of the prize contained inside. Selecting box $i$ results in a payoff of $v_i$ and immediately ends the search process. 
The searcher's goal is to design an adaptive policy (i.e., a choice of which operation to perform next, for every possible past history of operations and their outcomes) to maximize its expected utility, defined as the expectation of the prize selected, minus the sum of the inspection costs paid while opening boxes. 
Weitzman shows that in a model of the problem where acquiring a box is only allowed after opening it, referred to as the \emph{obligatory inspection} model, the optimal solution is nonadaptive and has a simple index-based structure. 

However, in many real-world environments such as hiring or school search, the agent can acquire a box (select an option) ``blind", i.e. without opening it and paying the inspection cost. Such scenarios motivate the \emph{nonobligatory inspection} model, introduced by Doval~\cite{Doval18}\footnote{A few papers \cite{GuhaMS08, ChangL09, AttiasKLS17, Doval18} have studied the same model in different contexts---see the related work section. \cite{Doval18} introduced the model in the context of search theory as a variant of Weitzman's model.}, where the searcher is allowed to acquire a box without opening it first. 
Prior literature presented evidence of  complexity of the optimal solution for Pandora's box problem with nonobligatory inspection.
In particular, Doval 
presents an example of a problem instance (Problem 3 in~\cite{Doval18}) with three boxes --- A, B, and C --- such that the optimal policy first opens box A, but the question of whether it subsequently opens box B before C or vice-versa depends on the value of the prize discovered inside box A, making the order of inspection adaptive. Furthermore, recently~\cite{FuLL22} showed that finding the optimal solution is NP-hard. It is even unknown whether the problem belong to class NP.

We study Pandora's box with nonobligatory inspection model and its optimal structure, and provide structural, complexity class, and approximation scheme results. In what follows, we overview our main results and techniques.

\subsection{Our Results}

\subsubsection{Structure of the Optimal Policy}

We show that despite the seemingly complicated nature of optimal policy, e.g., adaptive order of visiting boxes, and computational hardness, it has a simple structure. In fact, we show that there exists an optimal policy that follows only two different pre-determined orders and transitions from one to the other at most at one point.

\paragraph{A two-phase structure.} We prove that the optimal policy sets an initial ordering $\pi$, and a cutoff index $k$. It opens boxes one at a time according to this ordering until it either: (a) sees a sufficiently large value, in which case it concludes by using Weitzman's policy with obligatory inspection on the unopened boxes, or (b) reaches box $k$ without seeing a sufficiently large value, in which case it accepts box $k$ without inspection. Observe, for example, that this implies that there is just a single box $k$ that will ever be accepted without inspection.\footnote{Although the property that there is a unique box to be claimed closed has been shown previously by~\cite{GuhaMS08} for discrete and finite distributions, the two-phase structure is a novel contribution.} In other words, the optimal solution consists of two phases, where in each phase, the order of visiting boxes is pre-determined and nonadaptive. Whenever the maximum observed value, hereafter called the \emph{outside option} and denoted by $\oo$, exceeds the threshold, the policy switches to the second phase.

This result is summarized in the following statement, and also illustrated as~\Cref{algoTwoStageThreshold} in~\Cref{sec:opt}. The theorem is proved in \Cref{sec:opt}.

\begin{restatable}{theorem}{thmStructure}\label{thm:structure}
There exists an optimal policy specified by an ordering $\pi: [n] \rightarrow [n]$ of the $n$ boxes, a threshold $\tau: [n] \rightarrow \R$ for each index, and index $k$, where $0 \leq k \leq n$, such that while it has not terminated runs the following procedure for $j = 1, \ldots, k,$ sequentially. 
\begin{itemize}
    \item 
    If $j < k$ and if the maximum observed value is less than the next threshold, $\oo = \max_{1 \leq i < j}v_{\pi(i)} \leq \tau(j)$, then the policy will open box $\pi(j)$.
    \item
    If $j = k$ and if the maximum observed value is less than the next threshold, $\oo = \max_{1 \leq i < j}v_{\pi(i)} \leq \tau(j)$, then the policy will claim box $\pi(j)$ closed and terminate.
    \item 
    Otherwise, if $\oo > \tau(j)$, then run Weitzman's optimal policy with outside option $\oo$ on unobserved boxes $\pi(j), \pi(j+1), \ldots, \pi(n),$ and terminate.
\end{itemize}
\end{restatable}

This result identifies the possibilities of claiming a closed box and claiming the outside option as either-or alternatives when the searcher decides the next action. In the first phase, i.e., while the maximum observed value is below the threshold, the optimal policy ignores the outside option completely, acts as if it were $0$, and relies only the closed box $\pi(k)$ as an alternative to opening boxes. In the second phase, however, there is a possibility of reverting to the outside option and no possibility of claiming a closed box.

\paragraph{Computing thresholds.}
We prove given the ordering of the first phase, $\pi$, the thresholds $\tau$ are computed in polynomial time with the following procedure. The threshold for box $\pi(j)$ is set to minimum $\oo$, such that running Weitzman's algorithm for $\pi(j), \pi(j+1), \ldots, \pi(n)$ with outside option $\oo$ has the same utility as following the (optimal) policy for $\pi(j), \pi(j+1), \ldots, \pi(n)$ with outside option $0$.

\subsubsection{Complexity Class} 
Pandora's box with nonobligatory inspection problem has been known to belong to PSPACE. There has been no evidence that showed the problem is not PSPACE-complete~\cite{BeyhaghiK19}, and as shown by \cite{FuLL22}, the problem is at least NP-hard. The two-phase structure of the optimal policy implies that this problem belongs to NP, and therefore is NP-complete. As stated, given any order $\pi$, the thresholds $\tau$ can be computed in polynomial time, and therefore the utility with respect to the order is verifiable in polynomial time. The proof of~\Cref{cr:NP} can be found in~\Cref{app:intro}.

\begin{restatable}{corollary}{crNP} \label{cr:NP}
Pandora's box with nonobligatory inspection belongs to class NP, and is NP-complete.  
\end{restatable}

\subsubsection{PTAS}
We provide the first\footnote{Alongside with an independent concurrent work of Fu, Li, and Liu---see related work for more discussion.} polynomial time approximation scheme for Pandora's box with nonobligatory inspection problem. Prior to our work, the best approximation results were $0.8$ approximation~\cite{GuhaMS08} for discrete and finite distributions, and $0.8-\eps$ approximation for general distributions~\cite{FuLL22}. The proof of~\Cref{thm:ptas} can be found in~\Cref{sec:ptas}. 
\begin{restatable}{theorem}{THMPTAS}\label{thm:ptas}
There exists a PTAS for the  Pandora's box with nonobligatory inspection problem. 
\end{restatable}

\subsubsection{Tight Approximation for Committing Policies}
Committing policies were defined by~\cite{BeyhaghiK19} as a set of $n+1$ order-nonadaptive policies each with at most one fixed box that the agent can only claim it closed. They showed that the best of these policies provide a $1-1/e$ approximation of the optimum with a tight $0.8$ bound for two boxes. However, the same problem was studied sooner by Guha et al.~\cite{GuhaMS08} in the context of wireless networks. The main contribution of \cite{GuhaMS08} is a $0.8$ approximation to the Pandora's problem with nonobligatory inspection when the support of each box value distribution is discrete and finite. We prove the $0.8$ approximation for all (including continuous) distributions as a corollary of \Cref{thm:structure}. The proof of~\Cref{thm:comitting} can be found in~\Cref{app:intro}.

\begin{restatable}{theorem}{thmComitting} \label{thm:comitting}
At least one of the possible $n+1$ committing policies, achieves at least $0.8$ of the optimal utility for Pandora's box with nonobligatory inspection problem.
\end{restatable}

\subsection{Our Techniques}

\subsubsection{Optimal Structure}

We first consider a standard generalization of Pandora's box problem, where an outside option is given for free, and the searcher can select it at any point (as an alternative to selecting one of the boxes).\footnote{For the original problem, this outside option is initially set to $0$.} This generalization provides a unified format for the original problem and the subproblems. Then, we study the behavior of optimal searcher and the optimal expected utility, for any set of uninspected boxes, as a function of the outside option. Our key lemma (\Cref{lm:thresholds_exist}) proves that for any set of uninspected boxes, there is a threshold, such that for outside options above the threshold, the optimal policy never claims a closed box, and for outside options below the threshold, the optimal expected utility is constant. The constant optimal expected utility property implies that the optimal policy with any outside option below the threshold can just mimic the action of an optimal policy with outside option $0$. On the other hand, since having an outside option above the threshold coincides with not ever claiming a closed box, in this situation, the optimal policy can mimic the action of Weitzman's policy. Furthermore, we extract additional properties of the outside options, which imply that as the searcher inspects boxes and the outside option (maximum observed value) is updated, there is at most one point where the outside option switches from being below the threshold of uninspected boxes to above. Altogether, these structural properties conclude our main structural result, \Cref{thm:structure}.   

\subsubsection{PTAS} \label{subsubsec:intro:ptas}

As a consequence of \Cref{thm:structure} (and also by \cite{GuhaMS08} for discrete and finite distributions), there is an optimal policy that has at most one fixed box that it may claim closed. Therefore, based on which box the fixed one is (if any) we can limit the search to one of $n+1$ possible optimal policies.\footnote{Note that although this construction seems similar to committing policies~\cite{BeyhaghiK19}, in contrast, here the policies can be order-adaptive (similar to the two-phase optimal policy), and the fixed box may be opened or claimed closed.} In other words, we consider all the $n+1$ possibilities, find a PTAS for each, and output the one with the highest expected utility.

Our proof involves a novel reduction to a framework by~\cite{FuLX18}. We first overview the framework, how it is used for stochastic probing problems, and the challenges in tailoring it to our problem. We conclude by a summary of how we overcame the challenges and performed the reduction.

\cite{FuLX18} establishes a general framework for online stochastic problems and devises a PTAS for this general formulation. The stochastic dynamic program formulation in~\cite{FuLX18} models a general online probing setting, where there is a set of elements, and the agent's goal is to adaptively probe the elements to maximize the expected reward.\footnote{For a formal discussion of~\cite{FuLX18} framework, see \Cref{sec:ptas}.} Whenever the agent probes an element, they get an immediate reward, and their internal state is updated. At the end of the process, the agent also gets a final reward dependent on their internal state. This framework has been successfully applied to many stochastic probing problems, the most relevant to our problem being Probemax (choose $m$ elements to probe adaptively and get the maximum value among elements probed) and committed Pandora's box problem (similar to Pandora's problem with obligatory inspection, but elements are forfeited forever if not selected). These two problems share two critical aspects of Pandora's problem with nonobligatory inspection, respectively: 1) the agent gets the maximum value among all elements probed and 2) there is a cost of inspection. Although, this poses a reduction from our problem to~\cite{FuLX18} framework as a plausible approach, we face additional technical barriers not present in prior reductions for Probemax and committed Pandora's box. While resolving these technical barriers, we uncover additional structure for our problem that may be relevant beyond our specific PTAS reduction.

\paragraph{Challenge 1: negative terms reflecting costs.} We define the internal state to represent the best value (or an approximation of the value) that the agent has seen in the past. However, almost all previous problems that reduce to~\cite{FuLX18} that use internal states to represent element value do not have cost of inspection.\footnote{For different choices of defining the internal state, see \Cref{sec:ptas}. Although committed Pandora's cost problems, involves paying inspection costs, they do not use the internal state to represent the cost.} The framework requires the internal states to be supported on a set of constant size. This will necessitate a discretization of the values. The canonical way to discretize the values is to round them (up or down) to an approximate value. However, since the reward at each step is the difference between internal state and the cost incurred, $V_j - c$ (where $V_j$ denotes the internal state at step $j$, and $c$ is the cost), rounding values to a nearby approximate value may completely distort the difference, restraining us from a small multiplicative approximation loss.

\paragraph{Prior techniques for eliminating costs.} 
\cite{KleinbergWW16} introduce a reduction from Pandora's box with obligatory inspection to a maximization problem without costs. They also introduce a property of policies called \emph{non-exposure} and show that the optimal policy of the obligatory inspection variant satisfies it. Informally speaking, a policy is non-exposed if it selects any inspected box whose value is above the threshold of the box. In any non-exposed policy, whenever a box is selected the gain is equal to a \emph{virtual value} defined as a function of the revealed value and properties of the box.\footnote{For a formal discussion see \Cref{sec:model} and \Cref{sec:ptas}.} 
The insight from~\cite{KleinbergWW16} for removing cost from the expected utility function has been successfully utilized in~\cite{Segev021} to prove equivalence of Pandora's box with commitment and free order prophets. Also, in Pandora's box with nonobligatory inspection problem, previously~\cite{BeyhaghiK19} used ideas from~\cite{KleinbergWW16} to provide utility upper bound and additional structure for the problem.

\paragraph{Failure of previous techniques, and a new reduction.} Unfortunately, the optimal policy for Pandora's problem with nonobligatory inspection may not always be non-exposed (See \Cref{ex:nonexposed} in \Cref{app:intro}). However,
given our knowledge about the two-phase structure of the optimal policy, we draw parallels between our two-phase policy and the non-exposed policies, and introduce \emph{stage-non-exposed} policies. Basically, we argue although the optimal policy might not select a box when its value is above the threshold, the optimal policy will always enter phase two and gains its respective utility. It is easy to calculate the expected utility during and after the phase transition.

\paragraph{Challenge 2: discretizing values.} Recall that by \Cref{thm:structure}, our two-phase policy is determined by an order over the boxes and their thresholds. To define the internal states of~\cite{FuLX18} framework, after our  cost-elimination reduction, we need to discretize the values observed and the potential thresholds onto a $O(\poly(1/\eps))$ sized-support. We show that the optimal thresholds are fairly robust to minor changes and can be rounded down to a multiples of $\eps \cdot \OPT$ between $0$ and $\OPT/\eps$, where $\OPT$ is the expected optimal utility. However, discretizing the values proved more challenging. The standard way to discretize an element value $v_i$ is to truncate the value space at $\E[\max_i v_i]/\eps$ (the truncation at $\E[\max_i v_i]/\eps$ is esssential to ensure that the probability of the value of \emph{any} element being above the truncated upper limit is at most $\eps$), and then discretize the values into increments of $\eps \cdot \E[\max_i v_i]/\eps$. However, since there is a potentially super constant gap between optimal utility $\OPT$ and the expected maximum value $\E[\max_i v_i]$, the standard discretization methods do not work. I.e., discretizing the values into multiples of $\eps \cdot \E[\max_{i} v_i]$ is too coarse to generate meaningful approximation guarantees, and discretizing $v_i$ into multiples of $\eps \cdot \OPT$ will yield good approximation for the agent utility, but the resulting support will have a super constant size. We resolve this issue by taking advantage of contribution of $v_i$ in the utility formula and internal states of~\cite{FuLX18} framework. We conclude that although we cannot truncate the distibution to a constant multiple of $\OPT$, \emph{for any fixed order}, selecting only a constant support on this large range, and discretizing onto it has a limited loss.

\paragraph{Challenge 3: Dependence of the discretized support on inspection order.}
At this point, given a fixed order of boxes, we resolved how to discretize the values onto a subset of constant support (although within a large range), to preserve the agent's utility reasonably. The next challenge is that we do not know the optimal order, to be able to select the descritization support! To resolve this issue, we show there is a bounded number of discretization methods. First, we show we can bound the multiplicative gap between the optimal expected utility $\OPT$ and the expected maximum $\E[\max_{i} v_i]$ by $n$.\footnote{The reason is that for each $i$, $OPT \geq \E[v_i]$, because the optimal policy can claim any box closed, and $\E[\max_{i} v_i] \leq \sum \E[v_i]$.} Then, as we have mentioned before, the support can always be truncated at $\E[\max_{i} v_i]/\eps$. Thus, the number of distinct supports of constant size is bounded by $n^{O(\poly{(1/\eps)})}$; i.e., there are this many discretization methods. Therefore, as input to \cite{FuLX18} framework, we try all of these possibilities of discretization, run all the PTAS outputs (one for each discretization method), and use the discretization that resulted in the highest agent utility from the PTAS policy.

\subsection{Related Work}

\paragraph{Prior work.} Pandora's problem (with obligatory inspection) was first proposed and analyzed in~\cite{Weitzman79}, which shows that an elegant nonadaptive policy (which opens boxes in a pre-defined order with pre-defined thresholds, and selects the first box with value above its threshold) is optimal. \cite{KleinbergWW16} provide a new interpretation of the problem and study various applications. Since the introduction of Pandora's problem, multiple papers in different communities~\cite{GuhaMS08, ChangL09, AttiasKLS17, Doval18} independently introduced and studied a stochastic probing problem that is in essence equivalent to Pandora's problem, but with nonobligatory inspection. This variant is then further studied in~\cite{BeyhaghiK19, FuLL22}. We will overview the prior works that are most related to our work. 

\cite{Doval18} explicitly formulates the nonobligatory inspection problem as a generalization to the original Pandora's problem and shows that the optimal policy may have a complicated structure. In particular, unlike the original Pandora's problem, there exists distributions for which no nonadaptive policy is optimal. This inspired the theory community to work on approximation algorithms and hardness results, as well as developing other variants of Pandora's problem. In addition, she provides sufficient conditions on the parameters of the problem under which she characterizes the optimal policy.

\cite{GuhaMS08} focus on discrete and finite distributions, and provide a structural result showing that in the optimal policy, at most one box will ever be claimed closed. They also provide a $0.8$ approximately optimal solution. Due to the discrete nature of the environment, they focus on optimal decision trees, where each node in the tree represents the remaining unispected boxes and the maximum observed value (outside option).
For their structural result, they start with an arbitrary optimal policy, and replace subtrees with higher outside options by subtrees with lower outside options while maintaining optimality. In our structural result, we use a similar idea. In particular,
after we characterize the optimal utility as a function of outside option, our optimal policy mimics the action of an optimal policy with outside option $0$ in the constant part of the utility function.
However, in contrast to \cite{GuhaMS08}, our techniques work for general distributions, and we give an  explicit characterization of the optimal policy.

\cite{KleinbergWW16} provided an alternative proof for Pandora's problem by reducing it to a maximization problem without cost. This also helps them compute the expected utility from Weitzman's policy, which we make extensive use of. A more detailed discussion can be found in~\Cref{sec:model} and~\Cref{sec:ptas}.

\paragraph{Concurrent Work.}  
Concurrent and independent of our present work, Fu, Li, and Liu also obtain a PTAS for Pandora's problem with nonobligatory inspection.\footnote{We learned this through personal correspondence with the authors.} To the best of our knowledge, their concurrent work contains a structural result, and their proof for the PTAS contains some similar ideas (e.g. their work also uses the~\cite{FuLX18} framework,  and they use similar techniques with regard to discretizing the random variables). In addition, Fu, Li, and Liu prove that finding the optimal policy for the Pandora's problem with nonobligatory inspection is NP-hard. An initial manuscript of their paper~\cite{FuLL22} includes the hardness result as well as an improved approximation ratio for committing policies over~\cite{BeyhaghiK19}.

\paragraph{Additional Related Work.} Finally, there is a growing body of work that extends Pandora's box problem to various other settings, such as Pandora's box with additional order constraints~\cite{BoodaghiansFLL20}, with correlated values distribution~\cite{GTTZ20}, where the agent needs to commit to taking the box or forfeiting it forever at each step~\cite{FuLX18, Segev021}, where each box could be partially opened at a reduced cost~\cite{AouadJS20}, where each box could be inspected using different methods each at a different cost (a generalization of the nonobligatory inspection model)~\cite{Beyhaghi19}, where the cost of inspection model is generalized to various combinatorial optimization problems~\cite{Singla18}, etc. This recent trend illustrates a general community interest in exploring online decision problems that models cost of inspection. 

\subsection{Organization}
The rest of the paper is organized as follows. In \Cref{sec:model}, we introduce the model and provide preliminaries. In~\Cref{sec:opt}, we characterize the structure of the optimal policy and prove~\Cref{thm:structure}. In~\Cref{sec:ptas}, we provide a PTAS for Pandora's problem with non-obligatory inspection. In~\Cref{app:intro},~\Cref{app:opt} and~\Cref{app:ptas}, we provide missing proofs from~\Cref{sec:intro},~\Cref{sec:opt} and~\Cref{sec:ptas}, respectively.

\section{Model and Preliminaries}\label{sec:model}

An agent has a set of $n$ boxes. This set is denoted by $\M$. Box $i$, $1 \le i \le n$, contains a prize, $v_i$, distributed according to distribution $F_i(v_i)$ with expected value $\E v_i$. The support of the distribution of box $i$ is $\Theta_i$, and $\Theta = \cup \Theta_i$ is the union of all supports. Prizes inside boxes are independently distributed.
Box $i$ has inspection cost $c_i$. While $F_i$ and $c_i$ are known; $v_i$ is not.

The agent sequentially inspects boxes, and search is with recall. Given a set of uninspected boxes, $\U$, and a vector of realized sampled prizes, $v$, the agent decides whether to stop or to continue search; if she decides to continue search she decides which box in $\U$
to inspect next. If she decides to inspect box $i$, she pays cost $c_i$ to instantaneously learn her value $v_i$.
If she decides to stop search, she can choose to select whichever box she pleases, regardless of whether it is inspected or not.
We use $\inspecti$ as an indicator for box $i$ being inspected and $\allocsubi$ as an indicator for the agent obtaining box $i$. Since one box can be obtained, $\sum_i \allocsubi \le 1$. The agent is an expected utility maximizer, where utility, $u$, is defined as the value of the box selected minus the sum of inspection costs paid. Given $v$, the vector of realized sampled prizes, and the two vectors of indicator variables, $\allocsub{ }$ and $\inspect{ }$, respectively indicating which boxes were selected and inspected, we have:

\[u(v,\allocsub{ },\inspect{ }) = \sum_i (\allocsubi v_i - \inspecti c_i) .\]

An important variant of the problem, in which inspection is required was introduced and optimally solved by Weitzman~\cite{Weitzman79}. He showed that when $\allocsubi \leq  \inspecti$, an index-based policy is the optimal solution. In this policy, the agent inspects boxes in decreasing order of their indices, $\strike_i$, where
$\strike_i$ is the unique solution to
\[ {\E}_{v_i\sim F_i}\left[ (v_i - \strike_i)^+ \right] = c_i \]
and is also known as the reservation value of box $i$.
The search stops either when one of the realized values is above the
reservation value of every remaining uninspected box,
or when the agent has inspected all of the boxes. Kleinberg et al.~\cite{KleinbergWW16} develop a new interpretation of Weitzman's characterization. They introduce a family of random variables $\covered_i := \min \{ v_i, \strike_i \}$
defined for each box $i$. These random variables are used to reduce Pandora's problem with obligatory inspection to a problem without costs, and provide an upper bound on its optimal expected utility. They also introduced an important property of polices for the original Pandora's box problem called non-exposed, which they show that the Weitzman's policy satisfies and hence prove the upper bound is tight. We provide the definition and related statements below.
\begin{definition}\label{def:non-exposed}
    \cite{KleinbergWW16} A policy is \textit{non-exposed} if it is guaranteed to select any inspected box $i$ which have value $v_i > \sigma_i$. Namely, $(\inspecti - \allocsubi ) \cdot (v_i - \sigma_i)^+$ is always exactly equal to $0$.  
\end{definition}
\begin{lemma} \label{lem:perboxUtility}
    \cite{KleinbergWW16} For any policy that satisfies $\allocsubi \leq \inspecti$ pointwise, $\E[\allocsubi v_i - \inspecti c_i] \leq \E[\allocsubi \kappa_i]$, furthermore, this holds with equality for every box $i$ if and only if the policy is non-exposed. 
\end{lemma}
\begin{proposition} \label{pr:w_utility} \cite{KleinbergWW16}
  Weitzman's policy on boxes $1 \le i \le n$ with distributions $F_i$ and
  inspection costs $c_i$, achieves expected utility $\E[\max_i \covered_i]$;
  the expected utility of any other policy subject to obligatory inspection cannot exceed this bound.
\end{proposition}

In order to represent the internal states of Pandora's box problem, we consider a generalization, in which we are given a set of uninspected boxes $\U$ and the setting is exactly the same as the original problem, except that we are also given an \emph{outside option} $\oo$ for free. We denote this problem, i.e., Pandora's box problem with nonobligatory inspection for unispecteded boxes $\U$ and outside option $\oo$, by $\Pro(\U, \oo)$. Using the same notation, our original problem is $\Pro(\M, 0)$. Similarly, we denote the \emph{state} of the problem with the set of uninspected boxes $\U$ and the maximum observed value $\oo$ as $(\U, \oo)$. Due to this formulation we use outside option and maximum observed interchangeably and denote them by $\oo$.

Without loss of optimality, we only consider policies whose actions only depend on the set of unispected boxes and the maximum observed value (outside option). Also, when studying optimal policies, we consider those that are \emph{pointwise} optimal, i.e., optimal for any state $(\U, \oo)$ they reach, even those with probability $0$. We denote the optimal expected utility of problem $\Pro(\U, \oo)$ by $\OPT(\U, \oo)$. Furthermore, without loss of optimality, we focus on \emph{deterministic} policies.

For policy $A$ and current state $(\U, \oo)$ we define the following functions:
\begin{itemize}
    \item $I^{A}: (\U, \oo)  \rightarrow (\M \cup \{\perp\})$ outputs the index of the next box considered.
    \item $G^{A}: (\U, \oo)  \rightarrow \{\text{Open}, \text{Close}, \text{Stop}\}$ outputs the operation on the next box, where the operations include open the box, claim the box closed, or terminate the policy without probing.
    \item \emph{Action} $H^A(\U, \oo) := (I^{A}(\U, \oo), G^{A}(\U, \oo))$ indicates the next box and operation. An action is called \emph{terminal} if the operation $G^A(\U, \oo)$ is equal to $\text{Close}$ or $\text{Stop}$.
\end{itemize}

\begin{definition}[state transition]
For any policy $A$, we will use $ST_{A}(\U, \oo)$ to denote all valid state transitions from state $(\U, \oo)$ when using policy $A$. Formally, 
\begin{itemize}
    \item 
    when $G^A(\U, \oo) = \text{Open}$, $ST_{A}(\U, \oo) = \{(\U \setminus \{i\}, \oo') \mid i = I^A(\U, \oo), \exists v \in \Theta_{i}, \oo' = \max(a, v)\}$;
    \item 
    when $G^A(\U, \oo) = \text{Close}$, $ST_{A}(\U, \oo) = \{(\U \setminus \{i\}, \E\sq*{v_{i}}) \mid i = I^A(\U, \oo)\}$;
    \item 
    when $G^A(\U, \oo) = \text{Stop}$, $ST_{A}(\U, \oo) = \emptyset$. 
\end{itemize}
\end{definition}

\begin{definition}[plausible sequence of states]
We will call a sequence of states $(\U_0, \oo_0), (\U_1, \oo_1), \cdots, (\U_k, \oo_k)$ \textbf{plausible} for policy $A$ if $\forall j \in [k]$, $\paren*{(\U_j, \oo_j) \rightarrow (\U_{j+1}, \oo_{j+1})} \in ST_{A}(\U_j, \oo_j)$. 
\end{definition}

\begin{definition}[Reachable State]
For any policy $A$, we will use $RS_{A}(\U, \oo)$ to denote all states that are reachable by policy $A$ from state $(\U, \oo)$. Formally, a state $(\U', \oo') \in RS_{A}(\U, \oo)$ if and only if there exists a plausible sequences of states for $A$ that start at $(\U, \oo)$ and ends at $(\U', \oo')$. For the sake of simplicity, we will use $RS(A)$ to denote all states that are reachable by policy $A$ from state $(\M, 0)$. 
For instance, if policy $A$ opens box $i$ first, then for any $i' \in \M, i' \neq i$, $(\M \setminus \{i'\}, 0)$ is not reachable by policy $A$ from $(\M, 0)$ since $A$ must inspect $i$ as its first action.  
\end{definition} 

\begin{definition}[use a backup box]
We will say that a policy $A$ \emph{uses a backup box} for problem $\Pro(\U, \oo)$ if either $A$ claims a box closed up front (namely $G^{A}(\U, \oo) =$ Close), or there exists a state $(\U', \oo') \in ST_{A}(\U, \oo)$ such that $A$ uses a backup box for problem $\Pro(\U', \oo')$. 
\end{definition}

\section{Structure of the Optimal Policy} \label{sec:opt} 

The main contribution of this section is proving the two-phase structure of the optimal policy stated in \Cref{thm:structure}. First, we study the optimal expected utility as a function of the outside options. As an immediate observation, the optimal utility is an increasing function of the outside option; however, as we show, there is more structure to it. Specifically, in state $(\U, \oo)$, for any set of uninspected boxes $\U$, there exists a threshold $\tau(\U)$ such that the optimal utility for $\Pro(\U, \oo)$ is the same for any outside option $\oo$ that does not exceed the threshold, and is strictly higher for those exceeding the threshold. Furthermore, there is always a policy that uses a backup box when the outside option is below threshold, while no optimal policy uses a backup box when the outside option exceeds the threshold. Then, we show in any optimal policy of $\Pro(\M, 0)$, there is at most one transition point when before this point the outside option (current maximum observed value) is always below the threshold of the current uninspected boxes, and after the point, it is always above. Finally, using this structure, we show there exists an optimal policy that while the outside option is below the threshold, takes the next action as if the outside option were $0$, and after the transition point, follows Weitzman's policy, proving the structure of \Cref{thm:structure}.

Full proofs of the section are in \Cref{app:opt}.

\begin{restatable}{observation}{obsMonotone}\label{obs:monotone}
$\OPT(\U, \oo)$ is increasing in $\oo$. 
\end{restatable}

\begin{definition}\label{def:thresholds}[$\tau(\U)$, threshold for uninspected boxes]
With abuse of notation, let $\tau(\U) \geq 0$ be the value that satisfies the following properties if there exists an optimal policy of $\Pro(\U, 0)$ that uses a backup box with positive probability.
\begin{enumerate}
    \item There exists an optimal policy of $\Pro(\U, \oo)$ that uses a backup box with positive probability if $0 \leq \oo \leq \tau(\U)$, and there does not exist any optimal policy of $\Pro(\U,\oo)$ that uses a backup box if $\oo > \tau(\U)$.
    \item $\tau(\U) = \argmax_{\oo \in \R_{\geq 0}} \set*{\OPT(\U, \oo) = \OPT(\U,0)}$.
\end{enumerate}
If no optimal policy of $\Pro(\U, 0)$ uses a backup box with positive probability, let $\tau(\U) = \NEG$. For ease of notation we assume $0 > \NEG$. 
\end{definition}

\Cref{lm:thresholds_exist} asserts that for any set of boxes such a threshold exists.

\begin{restatable}{lemma}{lmThresholdsExist} \label{lm:thresholds_exist}
For each set of boxes $\U$, $\tau(\U)$, as defined in \Cref{def:thresholds}, exists. 
\end{restatable}

\begin{proof}
If there is no optimal policy that uses a backup box with positive probability for $\Pro(\U, 0)$, $\tau(\U) = \NEG$ and exists by definition. Therefore, for the remainder of the proof, we only focus on the case that there is an optimal policy for $\Pro(\U, 0)$ that uses a backup box.

The proof consists of two main steps. In the first step, we show that for any set of boxes $\U$, there exists a threshold $\tau(\U)$, such that for outside option $\oo > \tau(\U)$, no optimal policy for $\Pro(\U, \oo)$ uses a backup box with positive probability, and when $\oo \leq \tau(\U)$, $\OPT(\U, \oo) = \OPT(\U, 0)$. In the second step, we show that $\tau(\U)$ from the first step is equal to $\argmax_\oo \{\OPT(\U,\oo)=\OPT(\U,0)\}$, and there is an optimal policy using backup boxes with positive probability for outside option below the threshold.

The proof of the first step is by induction over the size of $\U$, the number of boxes in the problem. Let $\tau(\U)$ be the largest value such that an optimal policy with outside option $\tau(\U)$ uses a backup box. If there is a single box, this means that the optimal utility of $\Pro(\U,\tau(\U))$ is equal to the expected value of the box, which is equal to the no outside option scenario $\Pro(\U,0)$. Using \Cref{obs:monotone}, this concludes the base case of the induction. For $|\U|>1$, there are two possibilities. If an optimal policy of $(\U, \tau(\U))$ claims a closed box in the first step, the argument is similar to $|\U| =1$. Otherwise, if the optimal policy starts with opening box $i$ and observing value $v_i$, designing the optimal policy for the remainder of the boxes is equivalent to designing the optimal policy for the boxes other than $i$ with an outside option that is the maximum of $\tau(\U)$ and $v_i$ (Equality \ref{eq:1}). Since the optimal policy for $\Pro(\U, \tau(\U))$ uses a backup box, there exists some value $v'_i$ for which the subproblem (the problem for $\U \setminus \{i\}$) uses a backup box, implying $\tau(\U \setminus \{i\}) \geq \tau(\U)$. We split the utility into the two parts where the outside option is equal to $\tau(\U \setminus \{i\})$, i.e., $v_i \leq \tau(\U \setminus \{i\})$, and where it is equal to $v_i$, i.e., $v_i > \tau(\U \setminus \{i\})$ (Equality \ref{eq:2}). By induction hypothesis, the part where $v_i \leq \tau(\U \setminus \{i\})$, has optimal utility equal to the subproblem with outside option $0$ (Equality \ref{eq:3}).
The sum of the two parts equals to utility of the problem given the set of boxes $\U$, and no outside option, where the first action is opening box $i$, and the rest of the action follows an optimal policy for $\Pro(\U, v_i)$ (Equality \ref{eq:4}).
This constructs a policy for $\Pro(\U, 0)$ and has optimal utility at most $\OPT(\U, 0)$ (Inequality \ref{ineq:5}). By \Cref{obs:monotone}, the inequality is in fact an equality. This concludes the first step of the proof.
\begin{align}
    \OPT(\U, \tau(\U)) &= - c_i + \E_{v_i\sim F_i}\sq*{\OPT(\U \setminus \{i\}, \max\{\tau(\U), v_i\})} \label{eq:1}\\
    &= - c_i + \E\sq*{\OPT(\U \setminus \{i\}, \max\{\tau(\U), v_i\}) ~\Big\vert~ v_i \leq \tau(\U \setminus \{i\})}\cdot \Pr \sq*{v_i \leq \tau(\U \setminus \{i\})} \label{eq:2}\\
    &\quad \quad \:\:\,  + \E\sq*{\OPT(\U \setminus \{i\}, \max\{\tau(\U), v_i\}) ~\Big\vert~ v_i > \tau(\U \setminus \{i\})}\cdot \Pr \sq*{v_i > \tau(\U \setminus \{i\})} \nonumber \\
    &=  - c_i + \OPT(\U \setminus \{i\}, 0) \cdot \Pr \sq*{v_i \leq \tau(\U \setminus \{i\})} \label{eq:3}\\
    &\quad \quad \:\:\,  + \E\sq*{\OPT(\U \setminus \{i\}, v_i) ~\Big\vert~ v_i > \tau(\U \setminus \{i\})}\cdot \Pr \sq*{v_i > \tau(\U \setminus \{i\})} \nonumber \\
    &=  - c_i + \E_{v_i \sim F_i}[\OPT(\U \setminus \{i\}, v_i)]\label{eq:4}\\ 
    &\leq \OPT(\U, 0). \label{ineq:5}
\end{align}
Now, we move on to the second step of the proof. So far, we showed that there exists $\tau(\U)$ such that no optimal policy with strictly larger outside option uses a backup box; and all optimal policies with outside option below the threshold have the same utility. We first show for any nonnegative outside option $\oo'$ below the threshold, there exists an optimal policy that uses a backup box. This is straight-forward because $\OPT(\U, \oo') = \OPT(\U, 0)$ implies that following any optimal policy of $\OPT(\U, 0)$ is optimal for $\Pro(\U, \oo')$. Since we assumed there exists an optimal policy of $\Pro(\U, 0)$ that uses a backup box, there exists one that uses a backup box for any $\Pro(\U, \oo')$ where $0 \leq \oo \leq \tau(\U)$. Since all the problems with outside options $\oo' \geq 0$ satisfying $\OPT(\U, \oo') = \OPT(\U, 0)$ have an optimal policy that uses a backup box, for outside options $\oo'' \geq 0$ that no optimal policy uses a backup box, $\OPT(\U, \oo'') > \OPT(\U, 0)$. This concludes the proof. 

\end{proof}

The following lemma shows that in any optimal policy, the thresholds from \Cref{def:thresholds} for any set of uninspected boxes is such that once the maximum observed value exceeds the threshold at a stage, it always exceeds the thresholds at later stages.
\begin{restatable}{lemma}{lmOnceExceedsAlwaysExceeds}\label{lm:once_exceeds_always_exceeds}
Let $\OAL$ be an arbitrary optimal policy for problem $\Pro(\U, \oo)$ and let $(\U, \oo), (\U_1, \oo_1), \cdots, (\U_k, \oo_k)$ be any plausible sequence of states for $\OAL$. If $\oo > \tau(\U)$,  then $\oo_j > \tau(\U_j)$ for all $j \in [k]$.
\end{restatable}

\begin{proof}[Proof sketch.]
The proof is by contradiction. We show if $\oo_j < \tau(\U_j)$, then $\Pro(\U_j, \oo_j)$, and therefore, $\Pro(\U, \oo)$ have optimal policies that use backup boxes, which implies $\oo \leq \tau(\U)$. 
\end{proof}

The following lemma states that there exists an optimal policy that whenever the outside option (maximum observed value) exceeds the threshold of the unispected boxes, runs Weitzman's policy, and whenever the maximum observed value is less than the threshold of the unispected boxes takes the same action. 

\begin{restatable}{lemma}{lmLastStructuralLemma}\label{lm:last_structural_lemma}
There exists an optimal policy $\OAL$ for problem $\Pro(\M, 0)$ that satisfies the following: for any reachable state $(\U, \oo) \in RS(\OAL)$, 
\begin{itemize}
    \item 
    When $\oo \leq \tau(\U)$, $H^{\OAL}(\U, \oo) = H^{\OAL}(\U,0)$;
    \item 
    When $\oo > \tau(\U)$, $H^{\OAL}(\U, \oo) = H^{W}(\U, \oo)$. In fact, for any $(\U', \oo') \in RS_{W}(\U, \oo)$, $H^{\OAL}(\U', \oo') = H^{W}(\U', \oo')$, where $H^W(\U, \oo)$ represents the action Weitzman's policy would take given that $\U$ is the set of uninspected boxes and $\oo$ is the maximum value obtained so far by the algorithm.
\end{itemize}
\end{restatable}

\begin{proof}[Proof sketch.]
If no optimal policy uses a backup box for $\Pro(\M, 0)$ with positive probability, then Weitzman's policy is an optimal policy satisfying the statement. Note that for any reachable state $(\U, \oo)$ of Weitzman's policy, $\oo > \tau(\U)$, otherwise there is an optimal policy that claims a closed box, which is in contradiction with the initial assumption.  

Now, suppose there exists an optimal policy that uses a backup box for $\Pro(\M, 0)$. Let $(i_1, g_1), \cdots, (i_k, g_k)$ be the the first action of pointwise\footnote{As mentioned in \Cref{sec:model}, a policy is pointwise optimal if it is optimal for any reachable state, even those with probability $0$.} optimal deterministic policies for problems $\Pro(\U_0, 0)$, $\Pro(\U_1, 0)$, \ldots, $\Pro(\U_k, 0)$, respectively, where $\U_0 = \M$, $\U_j = \M \setminus \{i_1, \ldots, i_j\}$, and $k$ is the first time in the sequence, where the action taken, i.e., $g_k$, is terminal. Note that since for each problem in the sequence the outside option is $0$, claiming a closed box has at least as much utility as taking the outside option. Therefore, we assume $g_k = \text{Close}$.

The remaining step of the proof constructs $\OAL$ that follows the sequence of $(i_j, g_j)$ as long as the maximum observed value is below the threshold, and follows Weitzman's policy whenever it is above the threshold. Note that by \Cref{lm:thresholds_exist}, the optimal utility of $\Pro(\U_j, \oo_j)$ is equal to $\Pro(\U_j, 0)$ when $\oo_j \leq \tau(\U_j))$, and therefore following the optimal action for $\Pro(\U_j, 0)$ is also optimal for $\Pro(\U_j, \oo_j)$. Also, when $\oo_j > \tau(\U_j)$, no optimal policy uses a backup box with positive probability, and conditioned on not using a backup box, following Weitzman's policy is optimal. The formal discussion can be found in \Cref{app:opt}.
\end{proof}

\begin{proof}[Proof of \Cref{thm:structure}]
Let $\OAL$ be an optimal policy for the problem $\Pro(\M, 0)$ satisfying the conditions in \Cref{lm:last_structural_lemma}. If $\tau(\M) = \NEG$, by \Cref{lm:once_exceeds_always_exceeds} and \Cref{lm:last_structural_lemma}, $\OAL$ follows Weitzman's policy, implying the statement of the theorem. Now, suppose $\tau(\M) \geq 0$. By definition $\OAL$ uses a backup box. Let $(i_1, g_1), \cdots, (i_k, g_k)$ be the sequence of actions that $\OAL$ takes as long as the observed values are below the threshold (where by \Cref{lm:last_structural_lemma}, these observed values are assumed to be $0$). Without loss of optimality, we may assume that $g_k$ is the first time in this sequence that $\OAL$ claims a closed box and $g_1, \ldots, g_{k-1}$ correspond to opening boxes. This is trivial, since by assumption this sequence includes an action corresponding to claiming a closed box, all actions before claiming a closed box are opening boxes, and once a box is claimed closed the policy is at a terminal state. By  \Cref{lm:last_structural_lemma}, while the observed values are below the threshold, $\OAL$ takes $(i_1, g_1), \cdots (i_k, g_k)$. By \Cref{lm:once_exceeds_always_exceeds}, the maximum observed value at most at one point switches from being below the threshold to above the threshold, and once the maximum observed value is above the threshold, by  \Cref{lm:last_structural_lemma}, $\OAL$ follows Weitzman's policy.

We conclude by defining the parameters in the statement of the theorem. $k$ corresponds to the time where $\OAL$ claims a closed box when all observed values until that time were below their thresholds. $k$ is $0$ if no optimal policy uses a backup box. 
For $i \leq k$, $\pi(i)$ corresponds to the box visited at time $i$ by $\OAL$ if all the observed values were below their thresholds. Finally, $\tau(i) = \tau(\M \setminus \{\pi(1), \ldots, \pi(i-1)\})$, when $\tau(\M \setminus \{\pi(1), \ldots, \pi(i-1)\} \geq 0$, and is equal to a negative value, 
otherwise.
\end{proof}

Algorithmically, the optimal policy that satisfies the conditions in~\Cref{thm:structure} belongs to a class of policies that given an initial order and thresholds over the boxes, only switches its order of inspection (between the initial order provided and Weitzman's order) at most once. We term this class of polices two-phase policies, which is described in~\Cref{algoTwoStageThreshold}. (Notice that the thresholds $\tau_j$ could be negative.)
\begin{algorithm}\caption{Two-Phase Policy(InitialOrder=$i_1, \cdots, i_k, i^*$, Thresholds=$\tau_1, \cdots, \tau_{k}$)} \label{algoTwoStageThreshold}
\begin{algorithmic}[1]
\State{Let $\U_j = \M \setminus \{i_1, \cdots, i_{j}\}$.}
\For{$j = 1, \cdots, k$}
\State{Open box $i_j$, observe value $v_{i_j}$ from the box.}
\If{$v_{i_j} > \tau_j$}
    \State{Run Weitzman's policy on remaining boxes from state $(\U_j, v_{i_j})$.}
    \State{\algorithmicreturn}
\EndIf
\EndFor
\State{Claim box $i^*$ closed.}
\end{algorithmic}
\end{algorithm}
\section{PTAS}\label{sec:ptas}
In this section, we will present a PTAS for Pandora's problem with nonobligatory inspection (denoted as the problem $\PNOI := \PNOI(\M, 0)$). We will eventually reduce our problem to the general stochastic dynamic program formulation in~\cite{FuLX18}, but we need several intermediate steps to overcome difficulties caused by 1) our reward function having a negative cost term, and 2) the values of the boxes needing discretization. 
We will describe our reduction in the following order. In \Cref{subsec:FLX}, we will introduce the stochastic dynamic program formulation in~\cite{FuLX18} and its relevance to our problem. In \Cref{subsec:fix_backup}, we will reduce $\PNOI$ to its variant that fixes the unique box $i^*$ that may be claimed close, hereafter referred to as the \emph{backup box}. This variant, which we call $\PNOI_{i^*}$, enables us to focus on a fixed backup box for future reductions. In \Cref{subsec:index-threshold-seq}, we  introduce the notion of a pre-specified order threshold sequence that is relevant to all steps in our reduction. In \Cref{subsec:Pi*}, we focus on $\PNOI_{i^*}$ problem, and rephrase it using the new notion.
In \Cref{subsec:approxtwostage}, we will prove that the thresholds in the optimal two-phase policy are robust to additive perturbations. Hence, we can reduce the search space for the thresholds to $O(\frac{1}{\eps})$ without much loss in the utility. In \Cref{subsec:cost}, we will reduce the $\PNOI_{i^*}$ problem to a problem that always has \textit{nonnegative} reward at each step, which we call Tweaked $\PNOI_{i^*}$ (abbreviated as $\TPNOI_{i^*}$). This resolves our concern about the negative cost terms in our reward function. In \Cref{subsec:discretization}, we will discretize the $\TPNOI_{i^*}$ problem so that the value space of the system has constant support and call the resulting problem $\DTPNOI_{i^*}$. Finally, in \Cref{subsec:STDP}, we formulate the $\DTPNOI_{i^*}$ problem as the stochastic dynamic program ($\STDP_{i^*}$) specified in~\cite{FuLX18}, for which there exists a PTAS.

All the missing proofs of the section are in \Cref{app:ptas}.

%\subsection{Background and Additional Preliminaries} 

\subsection{The Stochastic Dynamic Program Formulation in~\cite{FuLX18}} \label{subsec:FLX} 
Here, we formally introduce the stochastic dynamic program, which is specified by a tuple $(\mathcal{V}, \mathcal{A}, f, g, h, n)$, and admits a PTAS with parameter $\eps$. We will also discuss several constraints on the parameters that are crucial to the existence of the PTAS (those text will be in italic).

\begin{itemize}
    \item 
    $\mathcal{V}$ describes the set of all possible internal values, \textit{which needs to be of a size that only depends on $\eps$}. 
    \item 
    $\mathcal{A} = \bigcup \mathcal{A}_i \cup \{\perp\}$ describes the action set, where $\mathcal{A}_i$ describe different ways to probe element $i$, and $\perp$ represents not probing anything. For each element $i$, \textit{$\mathcal{A}_i$ must be of a size that only depends number of elements and $\eps$ and is polynomial in the number of elements}. Moreover, the agent can never probe the same element twice (namely pick two actions from the same $\mathcal{A}_i$ set).  
    \item 
    $f$ describes how the value of the system changes from step $j$ to $j+1$. (i.e. The internal value at step $j+1$ is $V_{j+1} = f(V_{j}, a_j)$,where $a_j$ is the action at step $j$.)  \textit{The value of the system must be non-decreasing in $j$.}
    \item 
    $g(V_j, a_j)$ describes the immediate reward the agent gets at step $j$, given internal value $V_j$ and that the agent takes action $a_j$. Notice that $g$ can only depend on the value and action \textit{at} step $j$, but not the value and action \textit{before} step $j$. Furthermore, $g(V_j, a_j)$ can be stochastic but must have \textit{nonnegative expected value}.
    \item 
    Finally, $n$ represents the maximum steps the policy can take before terminating. $h(V_{n+1})$ describes the final additional reward at the end of the process, which depends on the value of the system before the policy terminates. $h(V_{n+1})$ must be \textit{pointwise nonnegative}. 
    \item 
    At the end of the process, the agent gets total reward $h(V_{n+1}) + \sum_{t=1}^n g(V_j, a_t)$. Here if the agent decides to terminate the process early at step $j^*$, we could view it as the agent taking a null action for all steps $j' > j^*$, and getting zero immediate rewards for those steps.
\end{itemize}

\subsection{Algorithmic Representation and Fixing Backup Box}\label{subsec:fix_backup} 
As we have seen in previous sections, the optimal policy (or at least there exists one that) is a two-phase policy described in \Cref{algoTwoStageThreshold} with some initial order and threshold $(i_1, \cdots, i_k, i^*, \tau_1, \cdots \tau_k)$, where $i^*$ is the unique box that may be claimed closed, hereafter referred to as the \emph{backup box}. In particular, when $\tau_1$ is negative, the two-phase policy does not use any backup box. In this case, the two-phase policy must be the Weitzman's policy. Otherwise, when the two-phase policy uses the backup box with non-zero probability, there are only $n = |\M|$ choices for the backup box. In this case, all of the $\tau_j$s (for $1 \leq j \leq k$) are nonnegative. To make our life easier in our reductions, we will mainly study a variant of the $\PNOI$ problem (which we will call $\PNOI_{i^*}$), where we are only allowed to claim a specific box $i^*$ closed without inspection. If for each $i^* \in \M$ we could find an approximately optimal policy $\ALG^{(i^*)}$ for problem $\PNOI_{i^*}$ with nonnegative thresholds, then simply taking the utility maximizing policy among $\ALG^{(i^*)}$ for each $i^* \in \M$ and the Weitzman's policy gives an approximately optimal policy for problem $\PNOI$. From now on, we will consider two-phase policies with a predetermined backup box $i^*$ and nonnegative thresholds $\tau_1, \cdots \tau_k$ (illustrated in \Cref{algoTwoStageThreshold2}). From this point on, we will use $\OPT := \OPT(\M, 0)$ to denote the optimal expected utility of problem $\PNOI := \PNOI(\M, 0)$. Similarly, we will use $\OPT_{i^*}$ to denote the optimal expected utility of problem $\PNOI_{i^*}$, which fixes the backup box $i^*$. 

\begin{algorithm}\caption{Two-Phase Policy with $i^*$ Backup ($\order=(i_1, \cdots, i_k, \tau_1, \cdots, \tau_{k})$)} \label{algoTwoStageThreshold2}
\begin{algorithmic}[1]
\State{Let $\U_j = \M \setminus \{i_1, \cdots, i_{j}\}$.}
\For{$j = 1, \cdots, k$}
\State{Open box $i_j$, observe value $v_{i_j}$ from the box.}
\If{$v_{i_j} > \tau_j$}
    \State{Run Weitzman's policy on remaining boxes from state $(\U_j, v_{i_j})$.}
    \State{\algorithmicreturn}
\EndIf
\EndFor
\State{Claim box $i^*$ closed.}
\end{algorithmic}
\end{algorithm}

\subsection{Index-Threshold Sequence, Classes of Policies, and Utilities }\label{subsec:index-threshold-seq} 

First, we introduce index-threshold sequence which is crucial for all the reduction steps and various classes of policies to be defined. Having fixed a backup box, $i^*$, and a position for the backup box in the order, $k+1$, the index-threshold sequence determines the boxes visited in order before the backup box and their respective thresholds.

\begin{definition}[Index-Threshold Sequence, $\order = (i_1, \cdots, i_k, \tau_1, \cdots \tau_k)$]
    We will define an \indexthreshold{} as an ordered sequence of box indices followed by an ordered sequence of threshold values of the same length. We will use $\order = (i_1, \cdots, i_k, \tau_1, \cdots \tau_k)$ to denote a specific \indexthreshold. 
\end{definition}

As we shall soon see in our reductions, for each problem $\GP \in \{\PNOI_{i^*}, \TPNOI_{i^*}, \DTPNOI_{i^*}, \STDP_{i^*}\}$, we will construct a class of policies $\class_{\GP}$ such that an \indexthreshold{} $\order = (i_1, \cdots, i_k, \tau_1, \cdots \tau_k)$ completely determines a specific policy with this class. Furthermore, there exists an optimal policy to problem $\GP$ that lies in the set $\class_{\GP}$. For instance, for the problem $\PNOI_{i^*}$, $\class_{\PNOI_{i^*}}$ would be the class of all two-phase policies with backup box $i^*$. $\class_{\TPNOI_{i^*}}, \class_{\DTPNOI_{i^*}}$ and $\class_{\STDP_{i^*}}$ will actually contain closely related policies to the two-phase policies. If a policy $\ALG$ belongs to the class of policies $\class_{\GP}$ and is determined by $\order$, we will say that $\ALG$ is parameterized with $\order$. 

We will also define a property of policies called \emph{below-threshold-nonadaptive} that holds for any policy in all policy classes $\class_{\GP}$ that we will define. Note that unlike two-phase property that specifies the action when a value exceeds the threshold (following Weitzman's policy), below-threshold-nonadaptive property is more general and does not specify the action in this case. This property only specifies the case where the values are below the thresholds and captures policies that are nonadaptive where the values are below the thresholds.
\begin{definition}[Below-Threshold-Nonadaptive] \label{def:belowThreshold}
A policy $\ALG$ parameterized with $\order = (i_1, \cdots, i_k, \tau_1, \cdots \tau_k)$ is \textit{below-threshold-nonadaptive} if 
\begin{enumerate}
    \item 
    $\ALG$ opens boxes in fixed order $i_1, i_2, \cdots, i_k$ while the value of none of the previously opened boxes have exceeded their thresholds.  
    \item 
    Given that before step $j$, the value of none of the previously opened boxes have exceeded their thresholds, $\ALG$'s expected utility from steps $\geq j$ is independent of what values it sees in steps $< j$. 
\end{enumerate}
\end{definition}

\begin{definition}[$U_{\GP}(\order)_{\geq j}$ and $U_{\GP}(\order)$]
Let $\ALG$ be the algorithm parametrized by $\order$ in class $\class_{\GP}$. We will now define $U_{\GP}(\order)_{\geq j}$ as the expected utility $\ALG$ gets at step $j$ from future steps, conditioned on the fact that in step $1, \cdots, j-1$, the value of the boxes are below the thresholds for the step. Since all polices we consider are below-threshold-nonadaptive, namely the utility of these policies are independent of previous values as long as they have not seen a box with above threshold value, $U_{\GP}(\order)_{\geq j}$ is well defined. We will use $U_{\GP}(\order)$ to denote the expected utility from $\order$ overall (namely, $U_{\GP}(\order) = U_{\GP}(\order)_{\geq 1}$).     
\end{definition}

We will make extensive use of these utility notations in our proofs, especially when comparing achievable utility between related problem formulations.

\begin{definition}
    For a set $\U$ of boxes and a fixed outside option $\oo$, we will define 
    \begin{align*}
    \weitz_{\U} := \max_{w \in \U} \kappa_w \quad \text{and} \quad \weitz_{\U}(\oo) := \max\set*{\max_{w \in \U} \kappa_w, \oo}.
\end{align*}
\end{definition}
Consequently, $\E[\weitz_{\U}]$ and $\E[\weitz_{\U}(\oo)]$ will be equal to the utility of Weitzman's policy (with no outside option) and that with an outside option $\oo$, respectively.

When analyzing the utility of a policy parameterized with $\order = (i_1, \cdots, i_k, \tau_1, \cdots, \tau_{k})$ at stage $j$ with uninspected boxes $\U_j = \M \setminus \{i_1, \cdots, i_{j-1}\}$, we use $\weitz_{\geq j}$ to represent $\weitz_{\U_j}$ and $\weitz_{\geq j}(\oo)$ to represent $\weitz_{\U_j}(\oo)$. 

\begin{note}[Expectations on $\weitz_{\geq j}$ and $\weitz_{\geq j}(\oo)$] When we take expectation over terms $\weitz_{\geq j}$ and $\weitz_{\geq j}(\oo)$, we will always take expectation over $v_w : w \in \U_j$, irrespective and independent of the range of $\oo$ we are taking expectation over. Hence, we will omit the subscript $v_j : j \in \U_j$ when taking expectations. E.g. when we use notation $\E\sq*{\weitz_{\geq j}}$, we mean $\E_{v_w : w \in \U_j}\sq*{\weitz_{\geq j}}$, and when we use notation $\E_{v > T}\sq*{\weitz_{\geq j}(v)}$, we mean $\E_{v > T, v_w : w \in \U_j} \sq*{\weitz_{\geq j}(v)}$. 
\end{note}

\subsection{$\PNOI_{i^*}$} \label{subsec:Pi*}
We will begin by defining $\class_{\PNOI_{i^*}}$, which will simply be the set of all two-phase policies with nonnegative thresholds. Recall that the two-phase policy (Algorithm~\ref{algoTwoStageThreshold}) for problem $\PNOI$ is determined by initial box order and thresholds $(i_1, \cdots, i_k, i^*, \tau_1, \cdots, \tau_k)$. Given that problem $\PNOI_{i^*}$ fixes the back up box, the class of two-phase policy $\PNOI_{i^*}$ is determined by $\order = (i_1, \cdots, i_k, \tau_1, \cdots, \tau_k)$. This proves validity of our choice of $\class_{\PNOI_{i^*}}$. 

We will also write out the utility recurrence formula for a two-phase policy $\ALG$ parameterized by $\order$ at step $j$. At step $j$, $\ALG$ inspects box $i_j$ and pays cost $c_{i_j}$. Then with probability $\Pr[v_{i_j} \leq \tau_j]$, the algorithm ignores the current value and transition to step $j+1$ in phase one. With probability $\Pr[v_{i_j} > \tau_j]$, the algorithm transitions into phase two and gets the same utility as Weitzman's policy would with outside option $v_{i_j}$. Hence we have the following recurrence: 
\begin{align*}
    U_{\PNOI_{i^*}}(\order)_{\geq j} &= \Pr[v_{i_j} \leq \tau_j] \cdot U_{\PNOI_{i^*}}(\order)_{\geq (j+1)} + \Pr[v_{i_j} > \tau_j] \cdot \E_{v_{i_j} > \tau_j}\sq*{\weitz_{\geq(j+1)}(v_{i_j})} - c_{i_j}.
\end{align*}
Finally, we have the following property for the optimal policy. 

\begin{restatable}{claim}{coroThresholdUtilityEq}
\label{coro:thresholdUtilityEq}
There exists an optimal two-phase policy $\ALG$ parametrized by $\order = \{i_1, \cdots, i_k, \tau_1, \cdots, \tau_k\}$ such that for all $j \in [k]$, $U_{\PNOI_{i^*}}(\order)_{\geq j} = \weitz_{\geq j}(\tau_j)$. 
\end{restatable}

\subsection{Discretizing Action (Threshold) Space} \label{subsec:approxtwostage}

To have a polynomial sized action space ($\mathcal{A}$), we need to discretize the thresholds. In this section, we will prove that the utility from the optimal \indexthreshold{} $\order$ is fairly robust to fluctuation in threshold values for the problem $\PNOI_{i^*}$.

Our first claim says that there exists an optimal two-phase policy for problem $\PNOI_{i^*}$ where all the thresholds $\tau_j$ are no larger than $\OPT$. This claim provides us with an upper bound to the search space for optimal thresholds. 
\begin{restatable}{claim}{claimTLessThanOPT}
\label{claim:TLessThanOPT}
For problem $\PNOI_{i^*}$, any optimal two-phase policy parametrized by $\order^* = (i_1, \cdots, i_k, \tau_1, \cdots \tau_k)$ that satisfies~\Cref{coro:thresholdUtilityEq} must satisfy for all $j \in [k]$, $\tau_j \leq \OPT$. 
\end{restatable}
Next, we prove that we can just search through \indexthreshold s with thresholds in increments of $\eps \cdot \OPT$, and find a good $\order$ whose associated two-phase policy gets at least $\OPT_{i^*} - \eps \cdot \OPT$ utility from $\PNOI_{i^*}$ problem. This enables us to restrict ourselves to considering thresholds of multiples of $\eps \cdot \OPT$ during our reductions in the next few sections. 
\begin{restatable}{proposition}{claimConstantT}
\label{claim:constantT}
Let $\order^* =(i_1, \cdots, i_{k}, \tau_1, \cdots, \tau_{k})$ be the parameter associated with an optimal two-phase policy for problem $\PNOI_{i^*}$ that satisfies Claim~\ref{coro:thresholdUtilityEq}. Then there exists another \indexthreshold{} $\order' = (i_1, \cdots, i_{k}, \tau_1', \cdots, \tau_{k}')$ with thresholds $\{\tau_j'\}_{j \in [k]}$ supported on $\W_L = \{0, \eps \cdot \OPT, 2 \cdot \eps \cdot \OPT, \cdots, \OPT\}$ such that $ U_{\PNOI_{i^*}}(\order') \geq \OPT_{i^*} - \eps \cdot \OPT$. 
\end{restatable}

\subsection{Removing Cost Terms  (Reducing $\PNOI_{i^*}$ to $\TPNOI_{i^*}$)} \label{subsec:cost}

In this section, we reduce problem $\PNOI_{i^*}$ to a problem with no costs $\TPNOI_{i^*}$. This step is helpful to have a finite internal value space $\mathcal{V}$ in our eventual reduction to~\cite{FuLX18} framework, while approximately preserving attainable utility.

\cite{FuLX18} requires the internal values to be supported on a set $\mathcal{V}$ with constant size. This will necessitate a discretization of the element\footnote{Elements in the stochastic dynamic program formulation correspond to boxes in our setting.} values, as those values are usually not supported on a small set. In various reductions to~\cite{FuLX18}, there are generally two ways to define the internal value $V_j$. The first option is to use $V_j$ to represent the best value (or an approximation of the value) that the agent has seen in the past. The second option is to use $V_j$ to represent the number of elements the policy has seen or selected. Given that in  Pandora's problem with nonobligatory inspection, the value the agent selects is very much dependent on all probed elements and not just a constant-size subset of elements, it is much more reasonable for us to use the first option -- use $V_j$ to represent some form of element value. 
However, almost all problems that reduce to~\cite{FuLX18} which use $V_j$ to represent element values do not have costs of inspection. The canonical way to discretize the values is to round the value up or down to an approximate value. However, if the reward is at step $t$ is $V_j - c$ for some cost $c$, then rounding $V_j$ to a nearby approximate value may completely distort the value of $V_j - c$ multiplicatively. 
To deal with this issue, we reduce the original $\PNOI_{i^*}$ problem to a problem without cost (we will call it $\TPNOI_{i^*})$ by drawing parallels between our two-phase policy and the non-exposed policy introduced by~\cite{KleinbergWW16} for the original Pandora's box problem. 

\subsubsection{Stage-Non-Exposed Policies}
\cite{KleinbergWW16} introduced the notion of  non-exposed policies (see \Cref{def:non-exposed}), which has been successfully applied to related problems with cost of inspection~\cite{Segev021}. Since optimal policy for $\PNOI$ (and hence $\PNOI_{i^*}$ for some $i^*$\footnote{Note that Weitzman's policy is non-exposed.}) may not always be non-exposed (See \Cref{ex:nonexposed} in \Cref{app:intro}), we provide a new related property that our policy satisfies.
 
Observe the following fact about non-exposed nonadaptive policies. 
\begin{definition} \label{def:nonadaptive}
    A nonadaptive policy is parametrized by $\order = (i_1, \cdots i_k, \tau_1, \cdots, \tau_k)$, and inspects boxes $i_1, \cdots i_k$ in sequential order. At step $j$, if $v_{i_j} > \tau_j$, then the policy selects box $i_j$ and terminates the process. 
\end{definition}

\begin{restatable}{claim}{claimNonAdaptiveNonExpose} \label{claim:nonAdaptiveNonExpose}
    \cite{Segev021} A nonadaptive policy parametrized by $\order = (i_1, \cdots i_k, \tau_1, \cdots, \tau_k)$ is non-exposed when $\tau_j \leq \sigma_{i_j}$. \footnote{This statement is almost an equivalence statement. A non-exposed policy with $\tau_j > \sigma_{i_j}$ must satisfy: $v_{i_j} \in (\sigma_{i_j}, \tau_j]$ with probability $0$. This can be formally dealt with easily.}
\end{restatable}

We can prove a claim with similar conditions to \Cref{claim:nonAdaptiveNonExpose} for two-phase policies despite the adaptivity of two-phase policies. We prove that although the optimal two-phase policy might not select a box when its value is above the threshold, the optimal policy \textit{will always enter phase two}. It is also easy to calculate the expected utility during and after the phase transition: if during the phase transition step $j$, the observed value is $v_{i_j}$, then the total utility from $\geq j$ step is just the expected utility from Weitzman's policy on remaining boxes with outside option $v_{i_j}$ minus the cost $c_{i_j}$.\footnote{From \Cref{sec:opt} we know that if a value is below threshold, the optimal mechanism can ignore it. Therefore considering $v_{i_j}$, the first value above the threshold, as the maximum observed value and therefore the outside option is valid.} This quantity is always at least $v_{i_j} - c_{i_j}$, the utility the agent would have gotten if they had just selected box $i_j$ and ended the process at step $j$. This gives us an alternative view of our two-phase policy: during the phase transition at step $j$, we immediately select box $j$ and get utility $v_{i_j} - c_{i_j}$, but we also get the ``leftover utility" from remaining boxes through Weitzman. This enables us to get rid of the cost term in similar manners to~\cite{Segev021}.

\begin{definition}[stage-non-exposed]
A two-phase policy with backup box $i^*$ parameterized by $\order = (i_1, \cdots, i_k, \tau_1, \cdots, \tau_k)$ is stage-non-exposed if for each $j \in [k]$, $\tau_j \leq \sigma_{i_j}$. 
\end{definition}

\begin{restatable}{claim}{claimTLessThanSigma} \label{claim:TLessThanSigma}
    For problem $\PNOI_{i^*}$, there exists an optimal two-phase policy parametrized by $\order = (i_1, \cdots, i_k, \tau_1, \cdots \tau_k)$ that is stage-non-exposed, namely, for each $j \in [k]$, $\tau_j \leq \sigma_{i_j}$.
\end{restatable}
\begin{corollary} \label{cor:tau_is_min_sigma_and_opt}
    For problem $\PNOI_{i^*}$, there exists an optimal two-phase policy parametrized by $\order = (i_1, \cdots, i_k, \tau_1, \cdots \tau_k)$ such that for each $j \in [k]$, $\tau_j \leq \min\{\sigma_{i_j}, \OPT\}$.
\end{corollary}

\begin{restatable}{proposition}{propNonExposedUtil} \label{prop:nonExposedUtil}
Let $\ALG$ be a stage-non-exposed two-phase policy parameterized by $\order = (i_1, \cdots, i_k, \tau_1, \cdots, \tau_k)$ and let $\phasetrans{j}$ denote whether the phase transition happens at step $j$. Then, $\ALG$ gets expected utility 
\begin{align*}
\sum_{j=1}^k \E\sq*{\phasetrans{j} \cdot 
 \paren*{\E_{v_{i_j} > \tau_j}[\kappa_{i_j}] + \E_{v_{i_j} > \tau_j}[(\weitz_{\geq j}- v_{i_j})^+]}} + \E\sq*{\paren*{1 - \sum_{j=1}^k  \phasetrans{j}} \cdot v_{i^*}}.
 \end{align*}
\end{restatable}

\Cref{prop:nonExposedUtil} gives rise to our problem formulation of $\TPNOI_{i^*}$, which given an \indexthreshold{} $\order = (i_1, \cdots, i_k, \tau_1, \cdots, \tau_k)$, computes the utility of the associated stage-non-exposed two-phase policy. 
\subsubsection{Formulation with No Cost}
\begin{mdframed}
Tweaked $\PNOI_{i^*}$ (we abbreviate as $\TPNOI_{i^*}$ ): 
\begin{itemize}
    \item
    Box set: $S = \M \setminus \{i^*\}$. Let $\U_j$ denote the remaining available item set at the beginning of each step $j$. 
    \item 
    In each step $j$, the agent can either open (with no repetition) a box $i_j$ and specify a priori a threshold $\tau_j \leq \sigma_{i_j}$, or choose to stop the process. If the value $v_{i_j}$ of box $i_j$ is at most $\tau_j$, then the agent gets $0$ reward. Otherwise the agent gets $ \E_{v_{i_j} > \tau_j}[\kappa_j] + \E_{v_{i_j} > \tau_j}[(\weitz_{\U_{j+1}} - v_{i_j})^+]$ reward, and the agent has to stop the process for the next round. 
    \item 
    When the agent decides to stop, if none of the boxes $i_j$ they have opened have value larger than their specified threshold $\tau_j$, then they get final reward $\E[v_{i^*}]$. Otherwise they get nothing when they stop.  
\end{itemize}
\end{mdframed}

We will now formally analyze the relationship between the utility of $\order$ from problems $\PNOI_{i^*}$ and $\TPNOI_{i^*}$.  
We define $\class_{\TPNOI_{i^*}}$ as the class of nonadaptive policies, which can be parametrized by $\order = (i_1, \cdots, i_k, \tau_1, \cdots, \tau_k)$. A nonadaptive policy parametrized by $\order$ opens boxes $i_1, \cdots, i_k$ in sequential order, until it sees a value $v_{i_j}$ above $\tau_j$, in which case it claims the reward and stops. If none of the boxes among $i_1, \cdots, i_k$ have value above the threshold, then the nonadaptive policy gets final reward $\E[v_{i^*}]$.
We first prove that $\class_{\TPNOI_{i^*}}$ contains the optimal policy, verifying that $\class_{\TPNOI_{i^*}}$ is well defined. 
\begin{restatable}{claim}{claimClassTPNOI} \label{claim:classTPNOI}
    $\class_{\TPNOI_{i^*}}$ contains an optimal policy for $\TPNOI_{i^*}$. 
\end{restatable}
We then verify that $\order$ induces the same utility for both $\PNOI_{i^*}$  problem (as a parameter to two-phase policy) and $\TPNOI_{i^*}$ problem (as a parameter to nonadaptive policy).
\begin{restatable}{proposition}{claimUtilityEqTPNOI} \label{claim:utilityEqTPNOI}
   Given any stage-non-exposed two-phase policy parameterized by $\order = (i_1, \cdots, i_k, \tau_1, \cdots, \tau_k)$, then $U_{\PNOI_{i^*}}(\order) = U_{\TPNOI_{i^*}}(\order)$. 
\end{restatable}
Finally, we combine threshold discretization from section~\ref{subsec:approxtwostage} and the equivalence of utility between $\TPNOI_{i^*}$ and $\PNOI_{i^*}$ in this section. 
\begin{restatable}{corollary}{claimApproxThresholdTPNOI}
\label{claim:approxThresholdTPNOI}
    There exists a stage-non-exposed two-phase policy parametrized by $\order = (i_1, \cdots, i_k, \tau_1, \cdots, \tau_k)$ where for all $j \in [k]$, $\tau_j \in \W_{L} = \{0, \eps \cdot \OPT, \cdots, \OPT\}$, such that 
    \begin{align*}
        U_{\TPNOI_{i^*}}(\order) = U_{\PNOI_{i^*}}(\order) \geq \OPT_{i^*} - \eps \cdot \OPT. 
    \end{align*}
\end{restatable}

\subsection{Discretization (Reducing $\TPNOI_{i^*}$  to $\DTPNOI_{i^*}$)} \label{subsec:discretization}

Currently, it is still not extremely clear how we would reduce from the $\TPNOI_{i^*}$ to a stochastic program. We will first briefly describe (without proof) how we could modify $\TPNOI_{i^*}$ into a problem that still has the same optimal utility, but whose reward functions are myopic (which is required by the stochastic dynamic program formulation). Observing this new and more adaptive formulation of $\TPNOI_{i^*}$(we call it Adaptive $\TPNOI_{i^*})$ will help us decide which values we need to discretize. 

\begin{mdframed}
Adaptive $\TPNOI_{i^*}$
\begin{itemize}
    \item
    Box set: $\U = \M \setminus \{i^*\}$. Let $\U_j$ denote the remaining available item set at the beginning of each step $j$. 
    \item 
    \textbf{During phase one}, in each step $j$, the agent can either open (with no repetition) a box $i_j$ and specify a priori a threshold $\tau_j \leq \sigma_{i_j}$, or choose to stop the process. If the value $v_{i_j}$ of box $i_j$ is at most $\tau_j$, then the agent gets $0$ reward. Otherwise \textbf{the agent gets $\kappa_{i_j}$ reward, update their internal value $V_j$ to $v_{i_j}$, and phase two starts.}
    \item 
    \textbf{During phase two, in each step $j$ the agent can open a box $i_j$ and get reward $(\kappa_{i_j} - v_{i_j})^+$. The agent update their internal value $V_{j}$ to $\max(\kappa_{i_j}, V_{j-1})$.}
    \item 
    When the agent decides to stop, if none of the boxes $i_j$ they have opened have value larger than their specified threshold $T_i$, then they get final reward $\E[v_{i^*}]$. Otherwise they get nothing when they stop.  
\end{itemize}
\end{mdframed}
In order for Adaptive $\TPNOI_{i^*}$ to be converted to a stochastic dynamic program (with format and constraints specified in~\Cref{subsec:FLX}, we need $V_j$ to have constant support and the number of choices of threshold $\tau_j$ for each box to be $O(n)$. We  have already seen in \Cref{subsec:approxtwostage} and \Cref{claim:approxThresholdTPNOI} that we could assume the thresholds are supported on $\W_{L} = \{0, \eps \cdot \OPT, 2 \cdot \eps \cdot \OPT, \cdots, \OPT\}$ with only additive $\eps \cdot \OPT$ loss to the attainable utility. The remaining challenge is to discretize the internal state $V_j$ onto a constant sized support. When the value $V_j$ is updated, it could either be updated to the value $v_{i_j}$ of a box $i_j$ during phase transition, or the value of $\kappa_{i_j}$ during phase two. Hence we need to discretize both of these quantities. 

Normally, for a problem without cost such as Probemax, the optimal expected utility from the agent is either above, or within a constant factor to $\E[\max_i v_i]$, the expected maximum value from the elements. In this case, the standard way to discretize an element value $v_i$ is to truncate the value space at $\E[\max_i v_i]/\eps$ (the truncation at $\E[\max_i v_i]/\eps$ is essential to ensure that the probability that the value of any element is above the truncated upper limit is at most $\eps$), and then discretize the values into increments of $\eps \cdot \E[\max_{i} v_i]$. Since the optimal agent utility is close to $\E[\max_{i} v_i]$, this rounding only affects the agent utility by $\eps$ factor. 

This method indeed works for discretizing our $\kappa_i$s. Since Weitzman's policy is a valid policy for the Pandora's box with nonobligatory inspection problem, it must be the case $\E[\max_{i} \kappa_{i_j}] \leq \OPT$. So the usual truncation plus discretization scheme works. 

However, for approximating $v_i$s as internal state values, the above scheme no longer works, since there is a potentially super constant gap between $\OPT$ and $\E[\max_i v_i]$. Discretizing the values into multiples of $\eps \cdot \E[\max_{i} v_i]$ is too coarse to generate meaningful approximation guarantees. On the other hand, discretizing $v_i$ into multiples of $\eps \cdot \OPT$ will yield good approximation for the agent utility, but the resulting support will have a super constant size. 

We resolve this issue by observing that $V_j := v_{i_j}$ actually only occurs once during phase change. Moreover, notice that the only effect of $v_{i_j}$ as an internal state is to compute the value of $\E_{v_{i_j} > \tau_j}[(\weitz_{\U_{j+1}} - v_{i_j})^+]$. Hence our discretization of $v_{i_j}$ (let's call it $\widetilde{v_{i_j}}$) doesn't need to be close to $v_{i_j}$ in \textit{value} at all. We just need the value of $\E_{v_{i_j} > \tau_j}[(\weitz_{\U_{j+1}} - \widetilde{v_{i_j}})^+]$ to be close to its original value. We further observe that although $v_{i_j}$ is hard to bound, the actual quantity that we need $\E[(\weitz_{\U_{j+1}} - \widetilde{v_{i_j}})^+]$ is bounded by $\OPT$. 

We will prove that given a fixed order of boxes $(i_1, \cdots, i_k)$, we can discretize $v_{i_j}$ onto a support $\W$, which contains multiples of $\eps^2 \cdot \OPT$ and is of size $O(\poly(1/\eps))$, such that $\{\E[(\weitz_{\U_{j+1}} - w)]\}_{w \in \W}$ covers $[0, \OPT]$ with granularity $\leq \eps^2 \cdot \OPT$. So \textit{given this fixed order of inspection}, discretizing $v_{i_j}$ onto $W$ will preserve the agent's utility reasonably. 

\begin{restatable}{proposition}{claimSupportExistence}\label{claim:supportExistence}
For any \indexthreshold{} $\order = (i_1, \cdots, i_k, \tau_1, \cdots \tau_j)$, there exists a support $\W$ of size $O(\poly(1/\eps))$ such that 
\begin{align*}
    \W \subseteq \set*{c \cdot \eps^2 \cdot \OPT: c \leq \frac{V_U}{\eps^2 \cdot \OPT}} \cup \{\infty\}, 
\end{align*}
where $V_U = \frac{\E[\max_i v_i]}{\eps}$, and for all $j \in [k]$ and for all $w \in R^+$, there exists a $w' \in \W, w' > w$ such that 
$\E[(\weitz_{\geq j} - w)^+] - \E[(\weitz_{\geq j} - w')^+] \leq \eps (1 - 2\eps) \cdot \OPT$. 
\end{restatable}

We will construct a $\W$ that satisfies ~\Cref{claim:supportExistence} in two steps.

\paragraph{Step one: partition the order into constant number of consecutive buckets. } For a consecutive bucket $B$ of boxes, we will define $f(B)$ and $l(B)$ to be the position of the first and last box in $B$ in a pre-specified ordered sequence of boxes. Now, given a particular \indexthreshold{} $\order = (i_1, \cdots, i_{k}, \tau_1, \cdots, \tau_{k})$, we will partition the ordered sequence of boxes into consecutive buckets $B_1, \cdots B_l$ such that each bucket $B_i$ contains the maximum number of boxes where it's still the case that $\E[\weitz_{\geq f(B_i)}] - \E[\weitz_{\geq l(B_i)}] \leq \eps^2 \cdot \OPT$. (Namely, $\Delta_{i} = \E[\weitz_{\geq f(B_i)}] - \E[\weitz_{\geq f(B_{i+1})}] > \eps^2 \cdot \OPT$). Notice that $l \cdot \eps^2 \cdot \OPT < \sum_{i=1}^{l} \Delta_{i} = \weitz_{\geq 1} \leq \OPT$. Thus $l$ is at most $\eps^2$. 

\paragraph{Step two: find constant support for $v_{i_j}$s for each bucket.} Next, for each bucket $B_i$, we will find a constant size support $0 = w_{i0} < w_{i1} < \cdots < w_{ip} = \E[\max_{i} v_i]/ \eps = V_{U}$ such that given $w_{i(j+1)}$, $w_{ij}$ is the smallest multiple of $\eps^2 \cdot \OPT$ such that
\begin{align*}
    \E[(\weitz_{\geq f(B_{i})} - w_{ij})^+] - \E[(\weitz_{\geq f(B_{i})} - w_{i(j+1)})^+] \leq \eps (1 - 3\eps) \cdot \OPT. 
\end{align*}
We will use $\W_{i}$ to denote the support set for bucket $i$, namely, $\W_{i} = \{w_{ij} : 0 \leq j \leq p\}$.

\begin{restatable}{claim} {claimVSupportWeitzDiff}\label{claim:vSupportWeitzDiff}
    For any $i_r$ in bucket $B_i$ (namely, when $f(B_i) \leq r \leq l(B_i)$), then 
    $$\E[(\weitz_{\geq r} - w_{ij})^+] - \E[(\weitz_{\geq r} - w_{i(j+1)})^+] \leq \eps (1 - 2\eps) \cdot \OPT.$$
\end{restatable}

The above claim enables us to prove \Cref{claim:supportExistence} by taking a union over the support we constructed for each bucket. The resulting support is denoted by $\W$. We now construct the Discrete $\TPNOI_{i^*}$ problem by rounding $\kappa$ \emph{down} to the next multiple of $\eps^2 \cdot \OPT$ and rounding $v$ \emph{up} to the nearest support in $\W$ when computing the agent's reward. 
\begin{mdframed}
Discrete $\TPNOI_{i^*}$ (we abbreviate as $\DTPNOI_{i^*}$): 
\begin{itemize}
    \item 
    Item set: $\U = \M \setminus \{i^*\}$. Let $\U_j$ denote the remaining available item set at the beginning of each step $j$. 
    \item 
    We will create a mapping for the continuous values to discrete values in $\W$: 
    \begin{itemize}
        \item $\widetilde{\kappa_i} = \floor{\frac{\kappa_i}{\eps^2 \cdot \OPT}} \cdot \eps^2 \cdot \OPT$ 
        \item $\widetilde{\weitz_{\U_i}} = \max_{w \in \U_i}\widetilde{\kappa_w}$
        \item 
        $\widetilde{v_i}$ is equal to the smallest support in $\W$ that has value at least $v_i$ (if $v_i > V_U$, then $\widetilde{v_i} = \infty$). 
    \end{itemize}
    \item 
    In each step $j$, the agent can either open (with no repetition) a box $i_j$ and specify a priori a threshold $\tau_j \in \W_{L}$ such that $\tau_j \leq \sigma_{i_j}$, or choose to stop the process. If the value $v_{i_j}$ of box $i_j$ is at most $\tau_j$, then the agent gets $0$ reward. Otherwise the agent gets $\E_{v_{i_j} > \tau_j}[(\widetilde{\weitz_{\U_{j+1}}} - \widetilde{v_{i_j}})^+] + \E_{v_{i_j} > \tau_j}[\widetilde{\kappa_{i_j}}]$ reward, and the agent has to stop the process for the next round.  
    \item 
    When the agent decides to stop, if none of the boxes $i_j$ they have opened have value larger than their specified threshold $\tau_j$, then they get final reward $\E[v_{i^*}]$. Otherwise they get nothing when they stop.  
\end{itemize}
\end{mdframed}
The exact same argument as in \Cref{claim:classTPNOI} shows that the optimal strategy for the $\DTPNOI_{i^*}$ problem is nonadaptive. Thus we can define $\class_{\DTPNOI_{i^*}}$ as a subset of $\class_{\TPNOI_{i^*}}$ (which only allows thresholds to be in $\W_{L} = \{0, \eps \cdot \OPT, \cdots, \OPT\}$).
It is clear that given a nonadaptive strategy $\ALG$ parametrized with \indexthreshold{} $\order \in \class_{\DTPNOI_{i^*}}$, $\ALG$ gets more expected reward from $\TPNOI_{i^*}$  compared to $\DTPNOI_{i^*}$ (because we round the positive terms, namely $\kappa_i$, downward and we round the negative terms, namely $v_i$, upwards). Now we will prove that the reward $\ALG$ gets from $\DTPNOI_{i^*}$ is within an additive $\eps \cdot \OPT$ away from the reward $\ALG$ gets from $\TPNOI_{i^*}$ . 
\begin{restatable}{proposition}{claimDTPNOI} \label{claim:DTPNOI}
    Given any nonadaptive policy $\ALG$ parametrized with $\order = (i_1, \cdots, i_k, \tau_1, \cdots, \tau_k)$. the expected reward $\ALG$ gets from $\DTPNOI_{i^*}$ is at least the expected reward $\ALG$ gets from $\TPNOI_{i^*}$ minus $2 \eps \cdot OPT$. Formally,
    \begin{align*}
        U_{\DTPNOI_{i^*}}(\order) \geq U_{\TPNOI_{i^*}}(\order) - 2\eps \cdot \OPT. 
    \end{align*}
\end{restatable}

\begin{restatable}{corollary}{claimDTPNOIopt} \label{claim:DTPNOIopt}
There exists a constant-size support $\W$ such that the optimal expected utility from $\DTPNOI_{i^*}$ is at least $\OPT_{i^*} - 3\eps \cdot \OPT$
\end{restatable}

At this point we have essentially established that if we know a near  optimal \indexthreshold{} for problem $\DTPNOI_{i^*}$, we can immediately find the optimal support $\W$ that satisfies~\Cref{claim:DTPNOIopt}. Unfortunately, we do not have such super power, as the near optimal solution is what we are trying to find in the first place! However, remember each $w \in \W$ must be at most $V_U = \frac{\E[\max_i v_i]}{\eps}$, and also $w$ must be a multiple of $\eps^2 \cdot \OPT$. We first extablish that the ratio between $\E[\max_{i} v_i]$ and $\OPT$ is at most $n$. Consequently, there are only polynomial number of possibilities for the choice of $w$. Since $\W$ is of size $\poly\paren*{\frac{1}{\eps}}$, there are only $O\paren*{n^{poly\paren*{\frac{1}{\eps}}}}$ many possible choices for $\W$.

\begin{restatable}{claim}{claimLessThanNOPT} \label{claim:lessThanNOPT}
$\E[\max_{i} v_i] \leq n \cdot \OPT$.
\end{restatable}
Hence once we provide a PTAS for the problem $\DTPNOI_{i^*}$, given a particular $\W$, we then can simply run the PTAS for all possible configurations of $\W$ and choose a $\W$ whose PTAS policy yields the maximum expected reward. By~\Cref{claim:DTPNOIopt}, this expected reward must be at least $(1 - \eps) \cdot (\OPT_{i^*} - 3\eps \cdot \OPT)$. 

\subsection{Obtaining the PTAS (Reducing $\DTPNOI_{i^*}$ to $\STDP_{i^*}$)} \label{subsec:STDP}

Now we are at the last step, which is to show that there exists a PTAS for the $\DTPNOI_{i^*}$ problem by reducing $\DTPNOI_{i^*}$ to the stochastic dynamic program with constant value space, whose format is defined in~\cite{FuLX18}. Based on a discretized version of Adaptive $\TPNOI_{i^*}$ (formulated in~\Cref{subsec:discretization}), our stochastic dynamic program is as follows. 

\begin{mdframed}
    Stochastic Dynamic Program (we abbreviate as  $\STDP_{i^*}$    ): 
    \begin{itemize}
        \item 
        $\M$ is the set of all boxes. As before, let box $i^*$ denote the backup box we have fixed.   
        \item 
        $n$ is maximum number of rounds, which is just $|\M|$. 
        \item 
        $\mathcal{V}$ is the set of all possible values of the system, which we will set as $\mathcal{V} = W$.  
        \item 
        $\mathcal{A}$ will represent the action space. Specifically, 
        For all $i \neq i^*$, let $\mathcal{A}_i = \{a_i^\tau\}_{\tau \in \W_L^i} \cup \{\infty\}$, where $a_i^T$ represent the action of opening box $i$ with threshold $T$, and $\W_L^i$ includes all elements in $\W_L = \{0, \eps \cdot \OPT, \cdots, \OPT\}$ that are at most $\sigma_{i}$. For the back up box $i^*$, let $\mathcal{A}_{i^*} = \{a_{i^*}^o\}$, where $a_{i^*}^o$ represent the action of opening the backup box without threshold. (The reward from claiming $i^*$ closed will be encoded in the final reward function $h$. ) 
        
        $\mathcal{A} = \bigcup_{i \in \M} \mathcal{A}_i$ is just the union over all possible actions for each box. 
        \item 
        $V_j$ represents the current value of the program, while $f$ is the transition function. Assume at step at we take an action for box $i_j \neq i^*$, then we will set transition function  
        \begin{align*}
            V_{j+1} = f(V_j, a_j = a_{i_j}^\tau) = \begin{cases}
                0 & \text{if $v_{i_j} \leq \tau$ and $V_j = 0$} \\
                \widetilde{v_{i_j}} & \text{if $v_{i_j} > \tau$ and $V_j = 0$}\\
                \max(\widetilde{\kappa_{i_j}}, V_j) & \text{if $V_j > 0$} 
              \end{cases}   
        \end{align*}
        For the action $b_{i^*}^o$, which opens the back up box, we will set 
        \begin{align*}
             V_{j+1} = f(V_j, a_j = a_{i^*}^o) = 
             \begin{cases}
                0 & \text{if $V_j = 0$}\\ \max(\widetilde{\kappa_{i^*}}, V_j) & \text{if $V_j > 0$}
             \end{cases}
        \end{align*}
        \item 
        We define the reward function $g$ for each state transition as for any $i_j \neq i^*$, 
        \begin{align*}
            g(V_j, a_j = a_{i_j}^\tau) = \begin{cases}
                0 & \text{if $v_{i_j}\leq \tau$ and $V_j = 0$}\\
                (\widetilde{\kappa_{i_j}} - V_j)^+ & \text{if $v_{i_j} > \tau$ or $V_j > 0$} 
            \end{cases}  
        \end{align*}
        Similarly, for $i_j = i^*$, 
        \begin{align*}
            g(V_j, a_j = a_{i^*}^{o}) = \begin{cases}
                0 & \text{if $V_j = 0$}\\
                (\widetilde{\kappa_{i^*}} - V_j)^+ & \text{if $V_j > 0$}
            \end{cases}
        \end{align*}
        \item 
        Finally, we define the final reward function as the reward the agent gets by claiming the backup box closed. This reward is not allowed if the threshold is crossed for some box. 
        \begin{align*}
            h(V_{n+1}) = \begin{cases}
                    \E[v_{i^*}] & \text{if $V_j = 0$}\\
                    0 & \text{if $V_j > 0$}
                \end{cases}
        \end{align*}
        \item 
        The total reward the agent gets is $\sum_{j=1}^n g(V_j, a_j) + h(V_{n+1})$. 
    \end{itemize}
\end{mdframed}
Firstly, as always, when we change the formulation of the problem, we need to define the class of policies that can be parameterized by $\order = (i_1, \cdots, i_k, \tau_1, \cdots, \tau_k)$. To do this, we first quickly observe that once $V_j > 0$, it is optimal to open all remaining boxes in problem $\STDP_{i^*}$. This makes sense, since we are deliberately trying to design the stochastic dynamic program so that in ``phase two" (once $V_j > 0$), the expected reward of the optimal policy correspond to the expected reward from $\DTPNOI_{i^*}$. 

\begin{restatable}{claim}{claimOPTSTDP} \label{claim:OPTSTDP}
    Once $V_j > 0$, it is optimal to open all remaining boxes, including the backup box.  
\end{restatable}

Using the same argument as in~\Cref{claim:classTPNOI} (essentially that policies that take deterministic actions given an internal value does as well as policies that can take randomized actions), we can show that there exists an optimal policy that is below-threshold-nonadaptive\footnote{Definition of below-threshold-nonadaptive is given in ~\Cref{def:belowThreshold}.}. Combining this with~\Cref{claim:OPTSTDP}, we conclude that there exists an optimal policy of the following form.

\begin{algorithm}\caption{$\STDP$ Policy($\order=(i_1, \cdots, i_k, \tau_1, \cdots, \tau_{k})$)} \label{algoSTDP}
\begin{algorithmic}[1]
\For{$j = 1, \cdots, k-1$}
\State{Probe element $i_j$ with threshold $\tau_j$, observe value $v_{i_j}$ from the box.}
\If{$v_{i_j} > \tau_j$}
    \State{Probe all remaining elements}
    \State{\algorithmicreturn}
\EndIf
\EndFor
\State{Get final reward $\E[v_{i^*}]$.}
\end{algorithmic}
\end{algorithm}

As with the previous subsections for class of policies, we define $\class_{\STDP_{i^*}}$ as the set of all valid $\STDP$ policies with $\tau_j$ supported on $\W_{L}$ and at most $\sigma_{i_j}$ for each $j$. This class is used to establish the following result.

\begin{restatable}{proposition} {claimAssociatedALG}\label{claim:associatedALG}
    Given a valid \indexthreshold{} $\order = (i_1, \cdots i_{k}, t_1, \cdots, t_{k})$, then 
    \begin{align*}
        U_{\DTPNOI_{i^*}}(\order) = U_{\STDP_{i^*}}(\order). 
    \end{align*}
\end{restatable} 

We are finally ready to prove the main theorem of the section, \Cref{thm:ptas}. The proof uses the relationship among $U_{\PNOI_{i^*}}, U_{\TPNOI_{i^*}}$, $U_{\DTPNOI_{i^*}}, U_{\STDP_{i^*}}$ and is provided in the appendix.

\bibliographystyle{alpha}
\bibliography{MasterBib}

\newcommand{\etalchar}[1]{$^{#1}$}
\begin{thebibliography}{CGT{\etalchar{+}}20}

\bibitem[AJS20]{AouadJS20}
Ali Aouad, Jingwei Ji, and Yaron Shaposhnik.
\newblock The pandora's box problem with sequential inspections.
\newblock {\em Available at SSRN 3726167}, 2020.

\bibitem[AKLS17]{AttiasKLS17}
Chen Attias, Robert Krauthgamer, Retsef Levi, and Yaron Shaposhnik.
\newblock Stochastic selection problems with testing.
\newblock {\em Available at SSRN 3076956}, 2017.

\bibitem[Bey19]{Beyhaghi19}
Hedyeh Beyhaghi.
\newblock {\em Approximately-optimal Mechanisms in Auction Design, Search
  Theory, and Matching Markets}.
\newblock Cornell University, 2019.

\bibitem[BFLL20]{BoodaghiansFLL20}
Shant Boodaghians, Federico Fusco, Philip Lazos, and Stefano Leonardi.
\newblock Pandora's box problem with order constraints.
\newblock In P{\'{e}}ter Bir{\'{o}}, Jason~D. Hartline, Michael Ostrovsky, and
  Ariel~D. Procaccia, editors, {\em {EC} '20: The 21st {ACM} Conference on
  Economics and Computation, Virtual Event, Hungary, July 13-17, 2020}, pages
  439--458. {ACM}, 2020.

\bibitem[BK19]{BeyhaghiK19}
Hedyeh Beyhaghi and Robert Kleinberg.
\newblock Pandora's problem with nonobligatory inspection.
\newblock In Anna Karlin, Nicole Immorlica, and Ramesh Johari, editors, {\em
  Proceedings of the 2019 {ACM} Conference on Economics and Computation, {EC}
  2019, Phoenix, AZ, USA, June 24-28, 2019}, pages 131--132. {ACM}, 2019.

\bibitem[CGT{\etalchar{+}}20]{GTTZ20}
Shuchi Chawla, Evangelia Gergatsouli, Yifeng Teng, Christos Tzamos, and Ruimin
  Zhang.
\newblock Pandora's box with correlations: Learning and approximation.
\newblock In Sandy Irani, editor, {\em 61st {IEEE} Annual Symposium on
  Foundations of Computer Science, {FOCS} 2020, Durham, NC, USA, November
  16-19, 2020}, pages 1214--1225. {IEEE}, 2020.

\bibitem[CL09]{ChangL09}
Nicholas~B. Chang and Mingyan Liu.
\newblock Optimal channel probing and transmission scheduling for opportunistic
  spectrum access.
\newblock {\em {IEEE/ACM} Trans. Netw.}, 17(6):1805--1818, 2009.

\bibitem[Dov18]{Doval18}
Laura Doval.
\newblock Whether or not to open pandora's box.
\newblock {\em Journal of Economic Theory}, 175:127--158, 2018.

\bibitem[FLL22]{FuLL22}
Hu~Fu, Jiawei Li, and Daogao Liu.
\newblock Pandora box problem with nonobligatory inspection: Hardness and
  improved approximation algorithms.
\newblock {\em arXiv preprint arXiv:2207.09545}, July 2022.

\bibitem[FLX18]{FuLX18}
Hao Fu, Jian Li, and Pan Xu.
\newblock A {PTAS} for a class of stochastic dynamic programs.
\newblock In Ioannis Chatzigiannakis, Christos Kaklamanis, D{\'{a}}niel Marx,
  and Donald Sannella, editors, {\em 45th International Colloquium on Automata,
  Languages, and Programming, {ICALP} 2018, July 9-13, 2018, Prague, Czech
  Republic}, volume 107 of {\em LIPIcs}, pages 56:1--56:14. Schloss Dagstuhl -
  Leibniz-Zentrum f{\"{u}}r Informatik, 2018.

\bibitem[GMS08]{GuhaMS08}
Sudipto Guha, Kamesh Munagala, and Saswati Sarkar.
\newblock Information acquisition and exploitation in multichannel wireless
  networks.
\newblock {\em CoRR}, abs/0804.1724, 2008.

\bibitem[KWW16]{KleinbergWW16}
Robert Kleinberg, Bo~Waggoner, and E.~Glen Weyl.
\newblock Descending price optimally coordinates search.
\newblock In {\em Proc.\ 17th ACM Conference on Economics and Computation
  (EC)}, pages 23--24, 2016.
\newblock arXiv:1603.07682 [cs.GT].

\bibitem[Sin18]{Singla18}
Sahil Singla.
\newblock The price of information in combinatorial optimization.
\newblock In Artur Czumaj, editor, {\em Proceedings of the Twenty-Ninth Annual
  {ACM-SIAM} Symposium on Discrete Algorithms, {SODA} 2018, New Orleans, LA,
  USA, January 7-10, 2018}, pages 2523--2532. {SIAM}, 2018.

\bibitem[SS21]{Segev021}
Danny Segev and Sahil Singla.
\newblock Efficient approximation schemes for stochastic probing and prophet
  problems.
\newblock In P{\'{e}}ter Bir{\'{o}}, Shuchi Chawla, and Federico Echenique,
  editors, {\em {EC} '21: The 22nd {ACM} Conference on Economics and
  Computation, Budapest, Hungary, July 18-23, 2021}, pages 793--794. {ACM},
  2021.

\bibitem[Wei79]{Weitzman79}
Martin~L Weitzman.
\newblock Optimal search for the best alternative.
\newblock {\em Econometrica: Journal of the Econometric Society}, pages
  641--654, 1979.

\end{thebibliography}
\appendix
\pagebreak

\section{Missing Proofs of \Cref{sec:intro}} \label{app:intro}

\crNP*
\begin{proof}[Proof of \Cref{cr:NP}]
We show that the decision version of the problem is in NP by providing an efficient verification algorithm for polynomial size certificates. The decision problem asks whether the optimal expected utility is above a given target.
Given a set of $n$ boxes $\M$, each with a distribution $F_i$, support $\Theta_i$, and cost $c_i$, as input, and initial order (inspection order for phase one) $\pi$, and a target expected utility $T$, as the certificate, we show in polynomial time how to verify whether an optimal two-phase policy using initial order $\pi$ generates expected utility that is at least $T$.

We try all the possibilities for $k$, and run the procedure for $k = 0, \ldots, n$. 
We consider $n$ sets $\U_0, \U_1, \ldots, \U_{n-1}$, where $\U_0 = \M$, and for $j \geq 1$, $\U_j = \M \setminus \{\pi(1), \ldots, \pi(j)\}$. First, for each value $\theta$ in the set of distribution supports $\Theta = \cup \Theta_i$, and sets $\U_j$ with $j \leq k$, we find Weitzman's optimal utility for uninspected set $\U_j$ and outside option $\theta$ and denote it as $W^{(\theta)}_j$. Since Weitzman's policy is an efficient algorithm~\cite{Weitzman79} this step is done in polynomial time. The next step is to find the optimal thresholds $\tau(j)$ and utilities $\OPT_j$ for all sets $\U_j$ with $j < k$ conditioned on using initial order $\pi$ and cutoff index $k$. We start from $j = k - 1$ and continue backwards to $j = 0$. Let $\OPT_{k-1} = \E[v_{\pi(k)}]$. We recursively define $\tau(j)$ and $\OPT_j$. Let $\tau(j) := \argmin_{\theta \in \Theta} W^{(\theta)}_j \geq \OPT_j$, and if no such $\theta$ exists, let $\tau(j) = \max \Theta + 1$. Let $\OPT_j = \E_{\theta_j \sim F_j} [W^{(\theta_j)}_{j+1} \one_{\theta_j \geq \tau(j+1)}] + \E_{\theta_j \sim F_j} [\OPT_{j+1} \one_{\theta_j < \tau(j+1)}]$. The optimal utilities are defined to reflect the best of continuing with phase one or switching to phase two. All the steps can be done in polynomial time. The final verification is comparing $\OPT_0$ and target utility $T$, returning true if $\OPT_0 \geq T$, and false otherwise.

\end{proof}

\thmComitting*

\begin{proof}[Proof of \Cref{thm:comitting}] The proof that we present here is a simplified version of proof of Theorem 5.3. in~\cite{GuhaMS08}, and follows similar logic. By \Cref{thm:structure} (and also by~\cite{GuhaMS08} for discrete and finite distributions), we know that there is at most one box that the optimal policy may claim closed. If the optimal policy does not have such a box, then Weitzman's policy is optimal implying the statement. Therefore, suppose that the optimal has a unique box $B$ that it may claim closed with probability $\beta$. We focus on two modified version of $A$. The first version, $A_\text{open}$, opens $B$ when $A$ claims $B$ closed and selects the maximum value observed. The second version, $A_\text{closed}$, claims $B$ closed whenever $A$ opens it. We first compare the utility of these two versions with the optimal expected utility $\OPT$ (utility of $A$). We show
\begin{align*}
    U(A_\text{open}) &\geq \OPT - \beta c_B;\\
    U(A_\text{closed}) &\geq \OPT - (1-\beta)(U(W)-c_B);
\end{align*}
where $U(W)$ is the expected utility of Weitzman's policy. The first inequality follows from $A_\text{open}$ paying extra cost of $c_B$ whenever $A$ claims $B$ closed, and receiving at least as much value compared to $A$. The second inequality follows from $A_\text{closed}$ not paying the cost of opening $B$ and losing at most the highest expected utility conditioned on not claiming any boxed closed which is $U(W)$. Multiplying the first inequality by $(1-\beta)$ and the second by $\beta$ gives:
\begin{align*}
   (1-\beta) U(A_\text{open}) + \beta U(A_\text{closed}) &\geq \OPT - \beta (1 - \beta) U(W).
\end{align*}
Note that since $A_\text{open}$ never claims a closed box, its utility is always at most that of Weitzman's policy. Also, since $A_\text{closed}$ never opens box $B$, its utility is at most that of the committing policy corresponding to $B$ (that never opens this box). Therefore, $U(A_\text{open})$ and $U(A_\text{closed})$ are both at most $\OPT^\text{commit}$, where $\OPT^\text{commit}$ is the maximum utility among committing policies. Upper bounding  $U(A_\text{open})$, $U(A_\text{closed})$, and $U(W)$ by $\OPT^\text{commit}$ gives
\begin{align*}
    \OPT^\text{commit}\left(1 + \beta (1-\beta) \right) &\geq \OPT,\\
    \OPT^\text{commit} &\geq \frac{1}{1+ \beta (1-\beta)}\OPT.    
\end{align*}
Since $0 \leq \beta \leq 1$, the minimum value for the right hand side occurs at $\beta = 1/2$, implying the statement.
\end{proof}

\begin{example}\label{ex:nonexposed}
This example shows that the optimal policy of Pandora's problem with nonobligatory inspection may not be non-exposed. Consider the following two boxes with $\eps$ being a sufficiently small number: 
\begin{itemize}
    \item 
    box $A$: $v_A = 
    \begin{cases}
        0 & \text{w.p. } 1/2 \\
        2 & \text{w.p. } 1/2
    \end{cases}$, $c_A = \eps$
    \item 
    box $B$: $v_B = 
    \begin{cases}
        0 & \text{w.p.} 1 - \eps\\
        \frac{1}{\eps} & \text{w.p. } \eps
    \end{cases}$, $c_B = \frac{1}{2}$
\end{itemize}
Based on the distributions and costs, $\sigma_A = 2 - 2\eps$ and $\sigma_B = 1/(2\eps)$. 

The optimal policy starts by opening box $A$. If $v_A = 0$, then it claims box $B$ closed. However, if $v_A = 2$, it continues with opening box $B$ and selecting $B$ if $v_B = 1/\eps$. Therefore, in this case, although $A$ has been inspected and $v_A > \sigma_A$, the optimal policy does not select it; which makes it an example of the optimal policy not satisfying non-exposure.
\end{example}

\section{Missing Proofs of \Cref{sec:opt}} \label{app:opt}

\obsMonotone*

\begin{proof}[Proof of \Cref{obs:monotone}]
For any $\oo, \oo'$ where $\oo' > \oo$, let $\OAL$ be an optimal policy for $\Pro(\U, \oo)$. We will show that there exists a policy for $\Pro(\U, \oo')$ which gets at least as much utility as $\OPT(\U, \oo)$. Consider a policy $\OAL'$ for $\Pro(\U, \oo')$ where it pretends the outside option is $\oo$ and at each stage does exactly what $\OAL$ would do conditioned on the revealed information. For any fixed sequence of values of the boxes, $\OAL'$ always pays the same costs as $\OAL$ and returns a value that is either equal to or greater than the value returned by $\OAL$.  
\end{proof}

The following lemma shows that given an optimal policy $\LOAL$ of a subproblem $\Pro(\U, \oo)$, we can construct another policy $A'$ such that for any reachable state in $RS_{\LOAL}(\U, \oo)$ follows $\LOAL$, and for any other reachable state follows $A$. This lemma is used in the proofs of \Cref{lm:once_exceeds_always_exceeds} and $\Cref{lm:last_structural_lemma}$.

\begin{lemma} 
\label{coroModifyOptimal}
Let $\LOAL$ be a policy that is optimal for the problem $\Pro(\U, \oo)$, then for any optimal policy $A$ for the problem $\Pro(\M, 0)$, then we can construct another optimal policy $A'$ such that for any state $(\U', \oo') \in RS(A')$, if $(\U', \oo') \in RS_{\LOAL}(\U, \oo)$, then $H^{A'}(\U', \oo') = H^{\LOAL}(\U', \oo')$. Otherwise, if $(\U', \oo') \in RS(A')$ but $(\U', \oo') \not \in RS_{\LOAL}(\U, \oo)$, then $H^{A'}(\U', \oo') = H^{A}(\U', \oo')$.
\end{lemma}
\begin{proof}
For any optimal policy $A$, let us construct the policy $A'$ such that at any state that is not a state in $RS_{\LOAL}(\U, \oo)$, $A'$ always takes the same action as $A$, however at a state $(\U', \oo') \in RS_{\LOAL}(\U, \oo)$, $A'$ will take the same action as $\LOAL$. We will first verify that this policy is valid (namely, $A'$ never reaches a state where the action at that state is ill-defined). To prove this, we will show that $RS(A') \subset RS(A) \cup RS_{\LOAL}(\U, \oo)$. For any state $(\U', \oo') \in RS(A')$ but $(\U', \oo') \not \in RS_{\LOAL}(\U, \oo)$, any sequence of states that start with $(\M, 0)$ and end with $(\U', \oo')$ that is plausible for policy $A'$ must not include the state $(\U, \oo)$ (otherwise since $(\U, \oo) \in RS_{\LOAL}(\U, \oo)$, $A'$ will take the same action as $\LOAL$ at $(\U, \oo)$, and similarly $A'$ will take the same action as $\LOAL$ at the next state, etc, until $A'$ reaches $(\U', \oo')$, therefore $(\U', \oo')$ must be reachable by $\LOAL$ from $(\U, \oo)$, which is a contraction). Thus $A'$ must take the same action as $A$ for all states this sequence of states, this means that this sequence of states is plausible for policy $A$ as well, which means that $(\U', \oo') \in RS(A)$. We conclude that $(\U', \oo')$ is either in $RS(A)$, or in $RS_{\LOAL}(\U, \oo)$. 

Now since $\LOAL$ is locally optimal at $(\U, \oo)$, for any $(\U', \oo') \in RS_{\LOAL}(\U, \oo)$, $H^{\LOAL}(\U', \oo')$ is the optimal first action for the problem $\Pro(\U', \oo')$. Similarly, for any $(\U', \oo') \in RS(A)$, $H^{A}(\U', \oo')$ is the optimal first action for the problem $\Pro(\U', \oo')$. We conclude that at any state $(\U', \oo') \in RS_{\LOAL}(\U, \oo)$, $H^{A'}(\U', \oo') = H^{\LOAL}(\U', \oo')$ is the locally optimal first action, and at any state $(\U', \oo') \in RS(A') \setminus RS_{\LOAL}(\U, \oo) \subseteq RS(A)$, $H^{A'}(\U', \oo') = H^{A}(\U', \oo')$ is the locally optimal first action as well. Hence $A'$ is optimal. 
\end{proof}

\lmOnceExceedsAlwaysExceeds*

\begin{proof}[Proof of \Cref{lm:once_exceeds_always_exceeds}]
We know by \Cref{def:thresholds} that if $\oo_0 > \tau(\U_0)$, then no optimal policy uses a backup box for the problem $\Pro(\U, \oo)$. Since $(\U, \oo) \in RS(\OAL)$, $\OAL$ is also optimal for the problem $\Pro(\U_0, \oo_0)$. Assume for contradiction that $\oo_i \leq \tau(\U_i)$ for some $i \in [k]$, then there exists an optimal policy $\LOAL$ that uses a backup box for the problem $\Pro(\U_i, \oo_i)$. By \Cref{coroModifyOptimal}, we know that there exists another optimal policy $\OAL'$ for problem $\Pro(\U, \oo)$ such that $\OAL'$ uses a backup box for the problem $\Pro(\U_i, \oo_i)$, and $(\U_i, \oo_i) \in RS_{\OAL'}(\U, \oo)$. 
Since $\OAL'$ uses a backup box for problem $\Pro(\U_i, \oo_i)$, and $(\U_i, \oo_i) \in RS_{\OAL'}(\U, \oo)$, $\OAL'$ must also use a backup box for problem $\Pro(\U, \oo)$, which is a contradiction to no optimal policy uses a backup box for the problem $\Pro(\U, \oo)$. 
\end{proof}

\lmLastStructuralLemma*

\begin{proof}[Proof of \Cref{lm:last_structural_lemma}]

If no optimal policy uses a backup box for $\Pro(\M, 0)$ with positive probability, then Weitzman's policy is an optimal policy satisfying the statement. Note that for any reachable state $(\U, \oo)$ of Weitzman's policy, $\oo > \tau(\U)$, otherwise there is an optimal policy that claims a closed box, which is in contradiction with the initial assumption.  

Now, suppose there exists an optimal policy that uses a backup box for $\Pro(\M, 0)$. Let $(i_1, g_1), \cdots, (i_k, g_k)$ be the the first action of pointwise\footnote{As mentioned in \Cref{sec:model}, a policy is pointwise optimal if it is optimal for any reachable state, even those with probability $0$.} optimal deterministic policies for problems $\Pro(\U_0, 0)$, $\Pro(\U_1, 0)$, \ldots, $\Pro(\U_k, 0)$, respectively, where $\U_0 = \M$, $\U_j = \M \setminus \{i_1, \ldots, i_j\}$, and $k$ is the first time in the sequence, where the action taken, i.e., $g_k$, is terminal. Note that since for each problem in the sequence the outside option is $0$, claiming a closed box has at least as much utility as taking the outside option. Therefore, we assume $g_k = \text{Close}$.

Now, consider a deterministic optimal policy $OAL$ that for problems $\Pro(\U_0, 0)$, $\Pro(\U_1, 0)$, \ldots, $\Pro(\U_k, 0)$, takes actions $(i_1, g_1), \cdots, (i_k, g_k)$ respectively. We show how to modify it to satisfy the conditions in the statement.

    Claim: If for some $h$ it is the case that for all $j \in [h]$ and for all $\oo_j \leq \tau(\U_j)$, $H^{\OAL}(\U_j, \oo_j) = H^{\OAL}(\U_j, 0)$, then there exists another optimal solution $\OAL'$ such that for all $j \in [h+1]$ and for all $\oo_j \leq \tau(\U_j)$, $H^{\OAL'}(\U_j, \oo_j) = H^{\OAL'}(\U_j, 0)$.

Proof of the claim: Firstly, observe that when for all $j \in [h]$ and for all $\oo_j \leq \tau(\U_j)$, $H^{\OAL}(\U_j, \oo_j) = H^{\OAL}(\U_j, 0)$, then conditioned on $\oo_j \leq \tau(\U): \forall j \in [h]$, the first $h$ actions $\OAL$ performs are exactly $(i_1, g_1), \cdots, (i_h, g_h)$. By \Cref{lm:once_exceeds_always_exceeds}, we know that for any $\oo_{h+1} \leq \tau(\U_{h+1})$ and any plausible sequence of states for $\OAL$ that starts with $(\M, 0)$ and ends at $(\U_{h+1}, \oo_{h+1})$, none of the intermediate states have their value exceed the threshold. This means that these intermediate states are exactly of form $(\U_j, \oo_j)$ where $\oo_j \leq \tau(\U_j)$. Since $\oo_{h+1} \leq \tau(\U_{h+1})$, we know that an optimal policy for the problem $\Pro(\U_{h+1}, 0)$ is also an optimal policy for the problem $\Pro(\U_1, \oo_{h+1})$. Now, let $\LOAL$ be an optimal policy for the problem $\Pro(\U_{h+1}, 0)$, and let $\LOAL(\oo_{h+1})$ be the corresponding policy for the problem $\Pro(\U_{h+1}, \oo_{h+1})$ that treats $\oo_{h+1}$ as value $0$. By \Cref{coroModifyOptimal}, we can construct another optimal strategy $\OAL'$, where $\OAL'$ takes the same action as $\OAL$, unless it is at a state in $ RS_{\LOAL(\oo_{h+1})}(\U_{h+1}, \oo_{h+1})$ for some $\oo_{h+1}$, in which case it will take the same action as $\LOAL(\oo_{h+1})$. Clearly, $(\U_1, \oo_1), \cdots (\U_h, \oo_h)$ are not reachable from $(\U_{h+1}, \oo_{h+1})$ for any $\oo_{h+1}$, hence $\OAL'$ takes the same action as $\OAL$ for those states. Furthermore, now $H^{\OAL'}(\U_{h+1}, \oo_{h+1}) = H^{\OAL'}(\U_{h+1}, 0)$ for any $\oo_{h+1} \leq \tau(\U)$. 

From a repeated application of the claim we have just proven, we know that from our original optimal policy $\OAL$, we can construct another optimal policy such that for all $j \in [k]$ and for all $\oo_j \leq \tau(\U_j)$, $H^{\OAL'}(\U_j, \oo_j) = H^{\OAL'}(\U_j, 0)$. This also implies that all reachable states $(\U, \oo) \in RS(\OAL)$ are of form $(\U_j, \oo_j)$ for some $j \in [k]$, hence the first condition in our lemma is satisfied. 

Now we will modify our optimal policy further so that the second condition in our lemma is satisfied. We know that for any $\U \subset \M$ and for any $\oo > \tau(\U)$, no optimal policy for $\Pro(\U, \oo)$ claims a box closed. Conditioned on not claiming any box closed, we know that an optimal policy for $\Pro(\U, \oo)$ is the Weitzman's algorithm.  Thus by Corollary~\ref{coroModifyOptimal}, we can modify $\OAL'$ and construct another algorithm $\OAL''$, where for any $(\U, \oo) \in RS(\OAL'')$ where $ \oo > \tau(\U)$, and for any $(\U', \oo') \in RS_{W}(\U, \oo)$, $H^{\OAL''}(\U, \oo) = W(\U, \oo)$. For any $(\U, \oo) \in RS(\OAL'')$ that does not satisfy our previous condition, $H^{\OAL''}(\U, \oo) = H^{\OAL'}(\U, \oo) = H^{\OAL}(\U, 0)$. By \Cref{lm:once_exceeds_always_exceeds}, we know that for any $(\U, \oo) \in RS(\OAL)$ where $ \oo > \tau(\U)$ and any $(\U', \oo') \in RS_{W}(\U, \oo)$, $ \oo' > \tau(\U')$. Hence $\OAL''$ takes the same action as $\OAL'$ for any reachable state $(\U, \oo)$ where $ \oo \leq \tau(\U)$. We conclude that $\OAL''$ satisfies our second condition, while still satisfying our first condition.

\end{proof}

\section{Missing Proofs of \Cref{sec:ptas}} \label{app:ptas}
The following claims use the fact that when % we fix values for
$\kappa_{w} : w \in \U$ are fixed, $\weitz_{\U}(\oo)$ can still be viewed as a function in $\oo$. 
\begin{claim} \label{claim:weitzSubadditive}
$\weitz_{\U}(\cdot): \R_{\geq 0} \rightarrow \R_{\geq 0}$ is a non-decreasing and subadditive function. 
\end{claim}

\begin{proof}%[proof of~\Cref{claim:weitzSubadditive}]
    Since $\weitz_{\U}(\oo)$ takes the max between a fixed number and $\oo$ (and then takes the expectation over the fixed number), $\weitz_{\U}(\oo)$ must be monotonically non-decreasing in $\oo$. Furthermore,
    \begin{align*}
        \weitz_{\U}(\oo + \beta) &= \max\set*{\max_{w \in \U} \kappa_w, \oo + \beta} \leq \sq*{\max\set*{\max_{w \in \U} \kappa_w, \oo}} + \beta \\
        &\leq \max\set*{\max_{w \in \U} \kappa_w, \oo} + \max\set*{\max_{w \in \U} \kappa_w, \beta} = \weitz_{\geq j}(\alpha) + \weitz_{\geq j}(\beta), 
    \end{align*}
    hence $\weitz_{\U}(\cdot)$ is also subadditive. 
\end{proof}

\begin{claim} \label{claim:weitzSubtract}
For any $\alpha, \beta \in \N$ such that $\alpha > \beta$, $\weitz_{\U}(\alpha) - \weitz_{\U}(\beta) \leq \alpha - \beta$. 
\end{claim}

\begin{proof}%[proof of~\Cref{claim:weitzSubtract}]
    For any $c \in \N$, $\max(c, \oo) - \max(c, \beta)$ is equal to $0$ when $\oo \leq c$, and is equal to $\oo - \max(c, \beta)$ when $\oo > c$. Both of these quantities are at most $\oo - \beta$. Let $c$ be $\max_{w \in \U} \kappa_w$ yields the claim. 
\end{proof}

%%%%%%%%%%%%
\coroThresholdUtilityEq*

\begin{proof}
    Essentially, at each step $j$ in a two-phase policy, the agent decides whether to move to phase two, or to forfeit the value $v_{i_j}$ forever and continue in phase one. Hence for any $\tau > \tau_{i_j}$, it must be at least as good to choose to continue to stage two, namely, 
        $U_{\PNOI_{i^*}}(\order)_{\geq j} \leq \weitz_{\geq j}(\tau)$.
    Changing the $\tau_j$s so that 
    they comply with the condition 
    $U_{\PNOI_{i^*}}(\order)_{\geq j}$ does not affect optimality. 
\end{proof}

\claimTLessThanOPT*

\begin{proof}
Let $\OAL$ parametrized by $\order^* = (i_1, \cdots, i_k, \tau_1, \cdots \tau_k)$ be any optimal two-phase policy for problem $\PNOI_{i^*}$ that satisfied~\Cref{coro:thresholdUtilityEq}. Notice that for any $j < k$, $U_{\PNOI_{i^*}}(\order^*)_{\geq j} \geq U_{\PNOI_{i^*}}(\order^*)_{\geq (j+1)}$. Moreover, the expected future utility from $\OAL$ in the first step is just $U_{\PNOI_{i^*}}(\order^*)_{\geq 1} = \OPT$. At step $j$, we know that $U_{\PNOI_{i^*}}(\order^*)_{\geq (j+1)} \leq \OPT$. Since by~\Cref{coro:thresholdUtilityEq} $U_{\PNOI_{i^*}}(\order^*)_{\geq (j+1)} = \weitz_{\geq (j+1)}(\tau_j)$, $\tau_j$ must also be at most $\OPT$. 
%(In fact since $\OAL$ is stage-non-exposed, $\tau_j \leq \min(\sigma_{i_j}, \OPT)$). 
\end{proof}

\claimConstantT*

\begin{proof}
Let $\order'$ have the same initial order $i_1, \cdots, i_k$ as $\order^*$, however, the thresholds in $\order'$ will be those in $\order^*$, but rounded down to a multiple of $\eps \cdot \OPT$. Namely, in $\order'$, for $j = 1, \cdots, k$, the threshold $\widetilde{\tau_j} = \eps \cdot \floor{\frac{\tau_j}{\eps}}$. Notice that since by Claim~\ref{claim:TLessThanOPT} $\tau_j \leq \OPT$, and $\widetilde{\tau_j} \leq \tau_j$, $\widetilde{\tau_j}$ is also at most $\OPT$. Hence $\widetilde{\tau_j} \in \W_{L} = \{0, \eps \cdot \OPT, 2 \cdot \eps \cdot \OPT, \cdots, \OPT\}$. 

We will now prove that $U_{\PNOI_{i^*}}(\order')_{\geq j} \geq U_{\PNOI_{i^*}}(\order^*)_{\geq j} - \eps \cdot \OPT$ for all $i \in [k]$ using induction.
    
We start off by assuming that for all $l > j$, $U_{\PNOI_{i^*}}(\order')_{\geq l} \geq U_{\PNOI_{i^*}}(\order^*)_{\geq l} - \eps \cdot \OPT$. In our base case where $j = k$, any two-phase policy just claims box $i_k$ closed in step $k$. Thus 
    \begin{align*}
        U_{\PNOI_{i^*}}(\order')_{\geq j} = \E[v_{i_k}] = U_{\PNOI_{i^*}}(\order^*)_{\geq j} \geq U_{\PNOI_{i^*}}(\order^*)_{\geq j} - \eps \cdot \OPT.
    \end{align*}  
        
Now when $j < k$, $\order^*$ and $\order'$ will open box $i_j$ with threshold $\tau_j$ and $\tilde{\tau_j}$ respectively. Hence
    \begin{align*}
        U_{\PNOI_{i^*}}(\order^*)_{\geq j} = \Pr[v_{i_j} \leq \tau_j] \cdot U_{\PNOI_{i^*}}(\order^*)_{\geq (j+1)} + \Pr[v_{i_j} > \tau_j] \cdot \E_{v_{i_j} > \tau_j}\sq*{\weitz_{\geq(j+1)}(v_{i_j})} - c_{i_j}
    \end{align*}
    and 
    \begin{align*}
        U_{\PNOI_{i^*}}(\order')_{\geq j} &= \Pr[v_{i_j} \leq  \widetilde{\tau_j}] \cdot U_{\PNOI_{i^*}}(\order')_{\geq(j+1)} + \Pr[v_{i_j} > \widetilde{\tau_j}] \cdot \E_{v_{i_j} > \widetilde{\tau_j}}\sq*{\weitz_{\geq(j+1)}(v_{i_j})} - c_{i_j}\\
        &= \Pr[v_{i_j} \leq  \widetilde{\tau_j}] \cdot U_{\PNOI_{i^*}}(\order')_{\geq (j+1)} + \Pr\sq*{\tau_{j} \geq v_{i_j} > \widetilde{\tau_j}} \cdot \E_{\tau_j \geq v_{i_j} > \widetilde{\tau_j}}\sq*{\weitz_{\geq(j+1)}(v_{i_j})} \\
        & \quad + \Pr[v_{i_j} > \tau_j] \cdot \E_{v_{i_j} > \tau_j}\sq*{\weitz_{\geq(j+1)}(v_{i_j})} - c_{i_j}.
    \end{align*}
    By the induction hypothesis, 
    \begin{align} \label{eq:one}
        U_{\PNOI_{i^*}}(\order')_{\geq (j+1)} \geq U_{\PNOI_{i^*}}(\order^*)_{\geq (j+1)} - \eps \cdot \OPT. 
    \end{align}
    Given that $\order^*$ is the parameter for an optimal two-phase policy that satisfies~\Cref{coro:thresholdUtilityEq}, we know that for any $\tau > \tau_{j+1}$, $\weitz_{\geq(j+1)}(\tau) \geq \OPT_{\geq (j+1)}$. By Claim~\ref{claim:weitzSubadditive}, $\weitz_{\geq (j+1)}(\cdot)$ is monotone and subadditive under addition. Therefore 
    \begin{align} \label{eq:two}
        \E_{\tau_j \geq v_{i_j} > \widetilde{\tau_{i_j}}}\sq*{\weitz_{\geq(j+1)}(v_{i_j})} &\geq \weitz_{\geq(j+1)}(\widetilde{\tau_{i_j}}) \\
        &\geq \weitz_{\geq(j+1)}(\tau_{i_j} - \eps) \geq  \weitz_{\geq(j+1)}(\tau_j) - \eps = \OPT_{\geq (j+1)} - \eps. 
    \end{align}
    By plugging in inequalities~(\ref{eq:one}) and (\ref{eq:two}) into our expansion of $U_{\PNOI_{i^*}}(\order')_{\geq j}$, we get   
    \begin{align*}
        U_{\PNOI_{i^*}}(\order')_{\geq j} &\geq \Pr[v_{i_j} \leq  \widetilde{\tau_{i_j}}] \cdot \paren*{U_{\PNOI_{i^*}}(\order^*)_{\geq (j+1)} - \eps \cdot \OPT} + \Pr\sq*{\tau_{i_j} \geq v_{i_j} > \widetilde{\tau_{i_j}}} \cdot \paren*{U_{\PNOI_{i^*}}(\order^*)_{\geq(j+1)} - \eps \cdot \OPT} \\
        &\quad + \Pr[v_{i_j} > \tau_{i_j}] \cdot \E_{v_{i_j} > \tau_{i_j}}\sq*{\weitz_{\geq(j+1)}(v_{i_j})} - c_{i_j}\\
        &\geq \paren*{\Pr[v_{i_j} \leq \tau_{i_j}] \cdot U_{\PNOI_{i^*}}(\order^*)_{\geq(j+1)} + \Pr[v_{i_j} > \tau_{i_j}] \cdot \E_{v_{i_j} > \tau_{i_j}}\sq*{\weitz_{\geq(j+1)}(v_{i_j})} - c_{i_j}} - \eps \cdot \OPT \\
        &= U_{\PNOI_{i^*}}(\order^*)_{\geq j} - \eps \cdot \OPT.  
    \end{align*}
    Finally, we conclude that the expected utility from $\order'$, which is equal to $U_{\PNOI_{i^*}}(\order')_{\geq 1}$, is at least $$U_{\PNOI_{i^*}}(\order^*)_{\geq 1} - \eps \cdot \OPT = \OPT_{i^*} - \eps \cdot \OPT.$$
\end{proof}

\claimTLessThanSigma*

\begin{proof}
    Let $\OAL$ be an optimal two-phase policy for $\PNOI_{i^*}$ parametrized by $\order = (i_1, \cdots, i_k, \tau_1, \cdots, \tau_k)$ that satisfies~\Cref{coro:thresholdUtilityEq}.  Since $\order$ is optimal, it must not be the case where removing an box from the order-threshold sequence improves utility. Therefore for any step $j < k$, $U_{\PNOI_{i^*}}(\order) \geq U_{\PNOI_{i^*}}(\order)_{ \geq (j+1)}$. 
    We can now expand the utility recurrence formula for two stage polices and get
    \begin{align*}
        U_{\PNOI_{i^*}}(\order)_{\geq j} &= \Pr[v_{i_j} \leq \tau_j] \cdot U_{\PNOI_{i^*}}(\order)_{\geq (j+1)} + \Pr[v_{i_j} > \tau_j] \cdot \E_{v_{i_j} > \tau_j}\sq*{\weitz_{\geq(j+1)}(v_{i_j})} - c_{i_j} \\
        &\geq U_{\PNOI_{i^*}}(\order)_{ \geq (j+1)}.
    \end{align*}
    By an exchange of terms, 
    \begin{align*}
        &\Pr[v_{i_j} > \tau_j] \cdot \E_{v_{i_j} > \tau_j}\sq*{\weitz_{\geq(j+1)}(v_{i_j})} - c_{i_j} > (1 - \Pr[v_{i_j} \leq  \tau_j]) \cdot U_{\PNOI_{i^*}}(\order)_{ \geq (j+1)}\\
        \Rightarrow & \Pr[v_{i_j} > \tau_j] \cdot \paren*{\E_{v_{i_j} > \tau_j}\sq*{\weitz_{\geq(j+1)}(v_{i_j})} - U_{\PNOI_{i^*}}(\OAL)_{ \geq (j+1)}} > c_{i_j}.
    \end{align*}
    By~\Cref{coro:thresholdUtilityEq}, it must be the case that $U_{\PNOI_{i^*}}(\order)_{\geq (j+1)} =\weitz_{\geq(j+1)}(\tau_{i_j})$; moreover, $c_{i_j} = \E[(v_{i_j} - \sigma_{i_j})^+]$ by definition. Hence 
    \begin{align*}
        \Pr[v_{i_j} > \tau_j] \cdot \paren*{\E_{v_{i_j} > \tau_j}\sq*{\weitz_{\geq(j+1)}(v_{i_j}) - \weitz_{\geq(j+1)}(\tau_{j})}} \geq \E[(v_{i_j} - \sigma_{i_j})^+].
    \end{align*}
    By Claim~\ref{claim:weitzSubtract}, $\weitz_{\geq(j+1)}(v_{i_j}) - \weitz_{\geq(j+1)}(\tau_{j}) \leq v_{i_j} - \tau_{j}$, thus 
    \begin{align*}
        \Pr[v_{i_j} > \tau_j] \cdot \E_{v_{i_j} > \tau_j}\sq*{v_{i_j} - \tau_j} &\geq \E[(v_{i_j} - \sigma_{i_j})^+]\\
        \Rightarrow  \E[(v_{i_j} - \tau_j)^+] &\geq \E[(v_{i_j} - \sigma_{i_j})^+]\\
    \end{align*}
    For all $j$, let $\tau_j' = \min(\tau_j, \sigma_{i_j})$. Then we could create another $\order' = (i_1, \cdots, i_k, \tau_1', \cdots, \tau_k')$ that is also optimal and satisfy conditions in the claim. 
\end{proof}

\begin{corollary} \label{cor:tau_is_min_sigma_and_opt}
    For problem $\PNOI_{i^*}$, there exists an optimal two-phase policy parametrized by $\order = (i_1, \cdots, i_k, \tau_1, \cdots \tau_k)$ such that for each $j \in [k]$, $\tau_j \leq \min\{\sigma_{i_j}, \OPT\}$.
\end{corollary}
\begin{proof}
By \Cref{coro:thresholdUtilityEq} and \Cref{claim:TLessThanSigma}. 
\end{proof} 

\propNonExposedUtil*

\begin{proof}
Firstly, we will again use the recurrence formula for two stage policy as well as expand the definition of $c_{i_j}$. 
    \begin{align*}
        U_{\PNOI_{i^*}}(\order)_{\geq j} &= \Pr[v_{i_j} \leq \tau_j] \cdot  U_{\PNOI_{i^*}}(\order)_{\geq (j+1)} + \Pr[v_{i_j} > \tau_j] \cdot \E_{v_{i_j} > \tau_j}[\weitz_{\geq (j+1)}(v_{i_j})] - c_{i_j} \\
        &= \Pr[v_{i_j} \leq \tau_j] \cdot  U_{\PNOI_{i^*}}(\order)_{\geq (j+1)} + \Pr[v_{i_j} > \tau_j] \cdot \E_{v_{i_j} > \tau_j}[\weitz_{\geq (j+1)}(v_{i_j}) - v_{i_j}]\\
        & \quad + \Pr[v_{i_j} > \tau_j] \cdot \E_{v_{i_j} > \tau_j}[v_{i_j}] - \E[(v_{i_j} - \sigma_{i_j})^+]\\
        &= \Pr[v_{i_j} \leq \tau_j] \cdot  U_{\PNOI_{i^*}}(\order)_{\geq (j+1)} + \Pr[v_{i_j} > \tau_j] \cdot \E_{v_{i_j} > \tau_j}[(\weitz_{\geq (j+1)}(0) - v_{i_j})^+] \\
        &\quad + \Pr[v_{i_j} > \tau_j] \cdot \tau_j + \E_{v_{i_j}}[(v_{i_j} - \tau_j)^+] - \E_{v_{i_j}}[(v_{i_j} - \sigma_{i_j})^+].
    \end{align*} 
Since $\ALG$ is a stage-non-exposed policy, for all $j \in [k]$, $\tau_j \leq \sigma_{i_j}$. Hence for any realized value of the random variable $v_{i_j}$, 
    \begin{align*}
        (v_{i_j} - \tau_j)^+ - (v_{i_j} - \sigma_{i_j})^+ = (\min(v_{i_j}, \sigma_{i_j}) - \tau_j)^+. 
    \end{align*}
    Taking the expectation over $v_{i_j}$ gives us 
    \begin{align*}
        \E_{v_{i_j}}[(v_{i_j} - \tau_{i_j})^+] - \E_{v_{i_j}}[(v_{i_j} - \sigma_{i_j})^+] = \E_{v_{i_j}}[(\min(v_{i_j}, \sigma_{i_j}) - \tau_j)^+]. 
    \end{align*}
    Thus 
    \begin{align*}
        \Pr[v_{i_j} > \tau_{i_j}] \cdot \tau_{i_j} + \E_{v_{i_j}}[(v_{i_j} - \tau_{i_j})^+] - \E_{v_{i_j}}[(v_{i_j} - \sigma_{i_j})^+] 
        &= \Pr[v_{i_j} > \tau_{i_j}] \cdot \tau_j + \E_{v_{i_j}}[(\min(v_{i_j}, \sigma_{i_j}) - \tau_j)^+]\\
        &= \Pr[v_{i_j} > \tau_j] \cdot \E_{v_{i_j} > \tau_j}[\min(v_{i_j}, \sigma_{i_j})]\\
        &= \Pr[v_{i_j} > \tau_j] \cdot \E_{v_{i_j} > \tau_j}[\kappa_{i_j}].
    \end{align*} 
    We can now rewrite the utility recurrence for $\ALG$ as 
    \begin{align*}
        U_{\PNOI_{i^*}}(\order)_{\geq j} = \Pr[v_{i_j} \leq \tau_j] \cdot U_{\PNOI_{i^*}}(\order)_{\geq (j+1)} + \Pr[v_{i_j} > \tau_j] \cdot \paren*{\E_{v_{i_j} > \tau_j}[(\weitz_{\geq (j+1)} - v_{i_j})^+] + \E_{v_{i_j} > \tau_j}[\kappa_{i_j}]}
    \end{align*}    
    Unrolling the recurrence gives the formula in the claim. 
\end{proof}

%%%%%%%%%%%%
\claimClassTPNOI*

\begin{proof}%[proof of \Cref{claim:classTPNOI}]
We prove that there is an optimal non adaptive solution for problem $\TPNOI_{i^*}$ by induction. Assume for any available item set $\U'$ where $|\U'| < |\U| = |\M \setminus \{i^*\}| = n-1$, there exists an optimal non adaptive solution to the tweaked problem. Observe that there must exist an optimal policy for the problem $\TPNOI_{i^*}$ with set $\U$ such that the first action is deterministic -- if the first action is randomized then that means there are two actions that are equally as good. Let $\OAL$ denote this this optimal deterministic policy. If the first action of $\OAL$ is to stop, then $\OAL$ is already non adaptive. On the other hand, if the first action of $\OAL$ is to open some box $i_1$ with threshold $\tau_1$. After the first step, either $v_{i_1} > \tau_1$ and the process stops, or $v_{i_1} \leq \tau_1$ and the agent still has $0$ reward. Thus a non adaptive optimal policy $\OAL_2$ for $\U \setminus \{i_1\}$ is also a locally optimal policy for the second case (where $v_{i_1} \leq \tau_1$). We can now device a new non adaptive optimal policy $\OAL'$ for tweaked problem on $\U$, where in the first step, $\OAL'$ opens box $i_1$ with threshold $\tau_1$, but in the case where $v_{i_1} \leq \tau_1$, $\OAL'$ takes future actions according to $\OAL_2$.

\end{proof}

\claimUtilityEqTPNOI*

\begin{proof}%[proof of \Cref{claim:utilityEqTPNOI}]
Given a stage-non-exposed two phase policy parametrized by $\order = (\i_1, \cdots, i_k, \tau_1, \cdots, \tau_k)$, then by~\Cref{prop:nonExposedUtil}, the utility recurrence 
\begin{align*}
    U_{\PNOI_{i^*}}(\order)_{\geq j} = \Pr[v_{i_j} \leq \tau_j] \cdot U_{\PNOI_{i^*}}(\order)_{\geq (j+1)} + \Pr[v_{i_j} > \tau_j] \cdot \paren*{\E_{v_{i_j} > \tau_j}[(\weitz_{\geq (j+1)} - v_{i_j})^+] + \E_{v_{i_j} > \tau_j}[\kappa_{i_j}]}.
\end{align*} 
Similarly, for a non adaptive policy parametrized by $\order$ for problem $\TPNOI_{i^*}$, at step $j$, the policy stops with probability $\Pr[v_{i_j} > \tau_j]$, in which case the agent gets reward $\paren*{\E_{v_{i_j} > \tau_j}[(\weitz_{\geq (j+1)} - v_{i_j})^+] + \E_{v_{i_j} > \tau_j}[\kappa_{i_j}]}$. Hence the utility recurrence for $\TPNOI_{i^*}$ is also 
\begin{align*}
    U_{\TPNOI_{i^*}}(\order)_{\geq j} = \Pr[v_{i_j} \leq \tau_j] \cdot U_{\TPNOI_{i^*}}(\order)_{\geq (j+1)} + \Pr[v_{i_j} > \tau_j] \cdot \paren*{\E_{v_{i_j} > \tau_j}[(\weitz_{\geq (j+1)} - v_{i_j})^+] + \E_{v_{i_j} > \tau_j}[\kappa_{i_j}]}.
\end{align*}
Moreover, at step $k+1$, $U_{\PNOI_{i^*}}(\order)_{\geq (k+1)} = U_{\TPNOI_{i^*}}(\order)_{\geq (k+1)} = \E[v_{i^*}]$. We conclude that $U_{\PNOI_{i^*}}(\order) = U_{\TPNOI_{i^*}}(\order)$. 
\end{proof}

\claimApproxThresholdTPNOI*

\begin{proof}%[proof of \Cref{claim:approxThresholdTPNOI}]
Let $\order = (i_1, \cdots, i_k, \tau_1, \cdots, \tau_k)$ be a parameter for a stage-non-expose two-phase policy, and let $\order' = (i_1, \cdots, i_k, \tau_1', \cdots, \tau_k')$, where $\tau_j' = \floor{\frac{\tau_j}{\eps \cdot \OPT}} \cdot \eps \cdot \OPT$. Then $\order'$ is also a parameter for a stage-non-expose two-phase policy, since we have only decreased the thresholds. By \Cref{claim:utilityEqTPNOI}, $U_{\PNOI_{i^*}}(\order') = U_{\TPNOI_{i^*}}(\order')$. By \Cref{claim:constantT}, $U_{\PNOI_{i^*}}(\order') \geq \OPT_{i^*} - \eps \cdot \OPT$. 
\end{proof}

\begin{claim} \label{claim:diffMonotone}
    For any random variable $X$ and $Y$, $\E[\max(X, Y)] - \E[Y] \geq  \E[(\max(X, Y) - T)^+] - \E[(Y-T)^+]$. 
\end{claim}
\begin{proof}
    Notice that $(\max(X, Y) - T)^+ - (Y - T)^+$ is only positive when $X > Y$ and $X > T$, in which case the term can be rewritten as $X - T - (Y - T)^+ = X - \max(Y, T)$. Thus $(\max(X, Y) - T)^+ - (Y - T)^+ = (X - \max(Y, T))^+$. We conclude that 
    \begin{align*}
         \E[(\max(X, Y) - T)^+] - \E[(Y - T)^+] = \E[(X - \max(Y, T))^+] \leq \E[(\max(X, Y) - Y)^+] = \E[\max(X, Y)] - \E[Y]. 
    \end{align*}
\end{proof}

%%%%%%%%%
\claimVSupportWeitzDiff*
\begin{proof}%[proof of \Cref{claim:vSupportWeitzDiff}]
    We know that $\weitz_{\geq r} = \max_{r' \geq r} \kappa_{r'}$ and $\weitz_{\geq f(B_i)} = \max_{r' \geq f(B_i)} \kappa_{r'}$, thus $$\weitz_{\geq f(B_i)} = \max \paren*{\weitz_{\geq r}, \max_{f(B_i) \leq r' < r} \kappa_{r'}}.$$ Then by ~\Cref{claim:diffMonotone}, we know that 
    \begin{align*}
       \E[(\weitz_{\geq f(B_i)} - w_{i(j+1)})^+] -  \E[(\weitz_{\geq r} - w_{i(j+1)})^+] \leq \E[(\weitz_{\geq f(B_i)}] - \E[(\weitz_{\geq r}] \leq \eps^2 \cdot \OPT.  
    \end{align*}
    Moreover, 
    \begin{align*}
       \E[(\weitz_{\geq r} - w_{ij})^+] \leq \E[(\weitz_{\geq f(B_i)} - w_{ij})^+].
    \end{align*}
    Thus 
    \begin{align*}
        \E[(\weitz_{\geq r} - w_{ij})^+] - \E[(\weitz_{\geq r} - w_{i(j+1)})^+] &\leq \E[(\weitz_{\geq f(B_i)} - w_{ij})^+] - \E[(\weitz_{\geq f(B_i)} - w_{i(j+1)})^+] + \eps^2 \cdot \OPT\\
        &\leq \eps (1 - 3\eps) \cdot \OPT + \eps^2 \cdot OPT = \eps (1 - 2\eps) \cdot \OPT. 
    \end{align*}
\end{proof}

\claimSupportExistence*
\begin{proof}%[Proof of ~\Cref{claim:supportExistence}]
We will now take the union of the support we found for each bucket $i \in [l]$ and also the low value range support to create the entire range of support.  
\begin{align*}
    \W = \paren*{\bigcup_{i=1}^p W_l} \cup \W_{L} \cup \{\infty\}.  
\end{align*} 
By construction, $\W$ contains only multiples of $\eps^2 \cdot \OPT$. By ~\Cref{claim:vSupportWeitzDiff}, for two nearest support $w < w'$ in $\W$, it must be the case that 
\begin{align*}
    \E[(\weitz_{\geq j} - w)^+] - \E[(\weitz_{\geq j} - w')^+] \leq \eps (1 - 2\eps) \cdot \OPT. 
\end{align*}
The only thing we need to verify is that $p+1 = |\W_i|$ is of constant size (specifically, $O\paren*{\frac{1}{\eps^2}}$) for all bucket $i$. Firstly, we observe that if we don't demand $w_{ij}$ to be a multiple of $\eps^2 \cdot \OPT$, it is trivial to construct a support set with size $O\paren*{\frac{1}{\eps^2}}$ such that 
\begin{align*}
    \E[(\weitz_{\geq f(B_{i})} - w_{ij})^+] - \E[(\weitz_{\geq f(B_{i})} - w_{i(j+1)})^+] \leq \eps (1 - 4\eps) \cdot \OPT.
\end{align*}
Then we can use a similar argument to ~\Cref{claim:vSupportWeitzDiff} to argue that rounding down the $w_{ij}$s onto the nearest multiple of $\eps^2 \cdot \OPT$ can only increase the Weitz term difference by $\eps^2 \cdot \OPT$, yielding the proposition.  
\end{proof}

%%%%%%%%%

\claimDTPNOI*

\begin{proof}%[proof of \Cref{claim:DTPNOI}]
    Obviously at step $k$ (the last step) the only action to take is to get reward $\E[v_{i^*}]$, which is the same between $\TPNOI_{i^*}$  and $\DTPNOI_{i^*}$. Now we will start by assuming by induction that 
    \begin{align*}
        U_{\DTPNOI_{i^*}}(\order)_{\geq (j+1)} \geq U_{\TPNOI_{i^*}}(\order)_{\geq (j+1)} - \eps \cdot OPT - \Pr[\max_{j' \in \U_{j+1}} v_{i_{j'}} > V_{U}]. 
    \end{align*}
    Writing out the reward recurrence for $\ALG$ formally for the problem $\DTPNOI_{i^*}$, 
    \begin{align*}
        U_{\DTPNOI_{i^*}}(\order)_{\geq j} = \Pr[v_{i_j} \leq \tau_j] \cdot U_{\DTPNOI_{i^*}}(\order)_{\geq (j+1)} + \Pr[v_{i_j} > \tau_j] \cdot \paren*{\E_{v_{i_j} > \tau_j}[(\widetilde{\weitz_{\geq (j+1)}} - \widetilde{v_{i_j}})^+] + \E_{v_{i_j} > \tau_j}[\widetilde{\kappa_{i_j}}]}. 
    \end{align*}
    By the fact that $\max(\cdot)$ is an submodular function, we know that $\widetilde{\weitz_{\geq (j+1)}} = \max_{j' \in \U_{j+1}} \widetilde{\kappa_{j'}} \geq \max_{j' \in \U_{j+1}} \set*{(1 - \eps^2) \cdot \kappa_{j'}} \geq (1 - \eps^2) \cdot \max_{j' \in \U_{j+1}} \kappa_{j'} = \weitz_{\geq (j+1)}$. For any $j < k$, let let $v_{i_j}^L, v_{i_j}^U$ be $v_{i_j}$ rounded down/up to the nearest support in $W$ respectively. Notice that there are two possibilities: either $v_{i_j}^U = \infty$ and $v_{i_j}^L = V_U$, or Claim~\ref{claim:vSupportWeitzDiff} tells us that for any $r < k$ and any fixed realization of $v_{i_j}$, 
    $\E[(\weitz_{\geq r} - v_{i_j}^L)^+] - \E[(\weitz_{\geq r} - v_{i_j}^U)^+] \leq \eps (1 - 2\eps) \cdot \OPT.$ Now we will split $\Pr[v_{i_j} > \tau_j] \cdot \E_{v_{i_j} > \tau_j}[(\widetilde{\weitz_{\geq j}} - \widetilde{v_{i_j}})^+]$ into two terms: 
    \begin{align*}
        &\Pr[v_{i_j} > \tau_j] \cdot \E_{v_{i_j} > \tau_j}[(\widetilde{\weitz_{\geq (j+1)}} - \widetilde{v_{i_j}})^+] \\
        = &\Pr[ V_{U} \geq v_{i_j} > \tau_j] \cdot \E_{V_{U} \geq v_{i_j} > \tau_j}[(\widetilde{\weitz_{\geq (j+1)}} - \widetilde{v_{i_j}})^+] + \Pr[v_{i_j} > V_U] \cdot  \E_{v_{i_j} > V_U}[(\widetilde{\weitz_{\geq (j+1)}} - \widetilde{v_{i_j}})^+]. 
    \end{align*}
    Now, since $\widetilde{v_{i_j}} = \infty$ when $v_{i_j} \geq V_U$, this means that $\E_{v_{i_j} > V_U}[(\widetilde{\weitz_{\geq (j+1)}} - \widetilde{v_{i_j}})^+] = 0$. Thus
    \begin{align*}
        \Pr[v_{i_j} > \tau_j] \cdot \E_{v_{i_j} > \tau_j}[(\widetilde{\weitz_{\geq (j+1)}} - \widetilde{v_{i_j}})^+] = \Pr[ V_{U} \geq v_{i_j} > \tau_j] \cdot \E_{V_{U} \geq v_{i_j} > \tau_j}[(\widetilde{\weitz_{\geq (j+1)}} - \widetilde{v_{i_j}})^+]. 
    \end{align*}
    Now, we will use the fact that $\E[(\weitz_{\geq r} - v_{i_j}^L)^+] - \E[(\weitz_{\geq r} - v_{i_j}^U)^+] \leq \eps (1 - 2\eps) \cdot \OPT$ for any $v_{i_j} \leq V_{U}$ to bound the difference between $\E_{V_{U} \geq v_{i_j} > \tau_j}[(\widetilde{\weitz_{\geq (j+1)}} - \widetilde{v_{i_j}})^+]$ and $\E_{V_{U} \geq v_{i_j} > \tau_j}[(\weitz_{\geq (j+1)} - v_{i_j})^+]$. 
    Firstly, for any fixed value of $v_{i_j}$ where $v_{i_j} \leq V_U$, 
    \begin{align*}
        \E[(\widetilde{\weitz_{\geq (j+1)}} - \widetilde{v_{i_j}})^+] =& \E[(\widetilde{\weitz_{\geq (j+1)}} - v_{i_j}^U)^+]\\
        \geq & \E[(\weitz_{\geq (j+1)} - v_{i_j}^U)^+] - \eps^2 \cdot OPT\\
        \geq &\E[(\weitz_{\geq (j+1)} - v_{i_j}^L)^+] -  \paren*{
            \E\sq*{(\weitz_{\geq (j+1)} - v_{i_j}^L)^+} - \E\sq*{(\weitz_{\geq (j+1)} - v_{i_j}^U)^+}
        } - \eps^2 \cdot OPT\\
        \geq & \E[(\weitz_{\geq (j+1)} - v_{i_j}^L)^+] - \eps (1- 2\eps) \cdot OPT - \eps^2 \cdot OPT\\
        = & \E[(\weitz_{\geq (j+1)} - v_{i_j}^L)^+] - \eps (1 - \eps)\cdot OPT.
    \end{align*}
    Thus when we take expectation over $V_U \geq v_{i_j} > \tau_j$, 
    \begin{align*}
        \E_{V_U \geq v_{i_j} > \tau_j}\sq*{(\widetilde{\weitz_{\geq (j+1)}} - \widetilde{v_{i_j}})^+} \geq \E_{V_U \geq v_{i_j} > \tau_j}[(\weitz_{\geq (j+1)} - v_{i_j}^L)^+] - \eps (1 - \eps)\cdot OPT. 
    \end{align*}
    From the definition of $\widetilde{\kappa_{i_j}}$, we know that $\widetilde{\kappa_{i_j}} \geq \kappa_{i_j} - \eps^2 \cdot OPT$. Let $\U_{j}$ be the set $\M \setminus \{i_1, \cdots, i_{j-1}\}$, 
    Now, we can finally bound $U_{\DTPNOI_{i^*}}(\order)_{\geq j)}$ from $U_{\TPNOI_{i^*}}(\order)_{\geq j}$ as follows. 
    \begin{align*}
        U_{\DTPNOI_{i^*}}(\order)_{\geq j} 
        &\geq \Pr[v_{i_j} \leq \tau_j] \cdot (U_{\TPNOI_{i^*}}(\order)_{\geq (j+1)} - \eps \cdot \OPT - \Pr[\max_{j' \in \U_{j+1}} v_{i_{j'}} > V_{U}] \cdot \OPT) \\
        &\quad + \Pr[V_U \geq v_{i_j} > \tau_j] \cdot \paren*{\E_{V_U \geq v_{i_j} > \tau_j}[(\weitz_{\geq (j+1)} - v_{i_j})^+] - \eps (1 - \eps) \cdot OPT} \\
        & \quad + \Pr[v_{i_j} > \tau_j] \cdot \paren*{\E_{v_{i_j} > \tau_j}[\kappa_{i_j}] - \eps^2 \cdot \OPT} \\
        &=  \Pr[v_{i_j} \leq \tau_j] \cdot U_{\TPNOI_{i^*}}(\order)_{\geq j} + \Pr[v_{i_j}> \tau_j] \cdot \paren*{\E_{v_{i_j} > \tau_j}[(\weitz_{\geq (j+1)} - v_{i_j})^+]  + \E_{v_{i_j} > \tau_j}[\kappa_{i_j}]} \\
        &\quad -\Pr[v_{i_j} > V_U] \cdot \E_{v_{i_j} > V_U}[(\weitz_{\geq (j+1)} - v_{i_j})^+] - \Pr[v_{i_j} \leq \tau_j] \cdot \Pr[\max_{j' \in \U_{j+1}} v_{i_{j'}} > V_{U}] \cdot \OPT- \eps \cdot \OPT\\
        &\geq U_{\TPNOI_{i^*}}(\order)_{\geq j} - \eps \cdot \OPT -  \Pr[v_{i_j} \leq \tau_j] \cdot \Pr[\max_{j' \in \U_{j+1}} v_{i_{j'}} > V_{U}] \cdot \OPT - \E_{v_{i_j} > V_U}[(\weitz_{\geq (j+1)} - v_{i_j})^+].
    \end{align*}
    Notice that Weitzman over a subset of boxes is a valid policy for the $\PNOI_{i^*}$ problem, hence $\E_{v_{i_j} > V_U}[(\weitz_{\geq (j+1)} - v_{i_j})^+] \leq \E[\weitz_{\geq (j+1)}]$ is at most $\OPT$. We now conclude that 
    \begin{align*}
        U_{\DTPNOI_{i^*}}(\order)_{\geq j} &\geq U_{\TPNOI_{i^*}}(\order)_{\geq j} - (\Pr[v_{i_j} > V_U] + \Pr[v_{i_j} \leq \tau_j] \cdot \Pr[\max_{j' \in \U_{j+1}} v_{i_{j'}} > V_{U}]) \cdot \OPT - \eps \cdot \OPT \\
        &\geq U_{\TPNOI_{i^*}}(\order)_{\geq j} - \Pr[\max_{j' \in \U_{j}} v_{i_{j'}} > V_{U}]) \cdot \OPT - \eps \cdot \OPT. 
    \end{align*}
   
    We conclude that at step $1$,
    \begin{align*}
        U_{\DTPNOI_{i^*}}(\order) \geq U_{\TPNOI_{i^*}}(\order) - \Pr[\max_{j \in \M} v_{j} > V_U] \cdot OPT - \eps \cdot \OPT. 
    \end{align*}
    Given that $V_{U} = \frac{\E[\max_{j \in \M} v_j]}{\eps}$, the probability that $\max_{j \in \M} v_{j}$ is great than $V_{U}$ is at most $\eps$. Hence 
    \begin{align*}
        U_{\DTPNOI_{i^*}}(\order) \geq U_{\TPNOI_{i^*}}(\order) - 2 \eps \cdot \OPT. 
    \end{align*}
\end{proof}

\claimDTPNOIopt*
\begin{proof}
    By ~\Cref{claim:approxThresholdTPNOI}, there exists an $\order = (i_1, \cdots, i_k, \tau_1, \cdots, \tau_k)$ where for all $j \in [k]$, $\tau_j \in \W_{L} = \{0, \eps \cdot \OPT, \cdots, \OPT\}$, such that 
    \begin{align*}
        U_{\TPNOI_{i^*}}(\order) = U_{\PNOI_{i^*}}(\order) \geq \OPT_{i^*} - \eps \cdot \OPT. 
    \end{align*}
    This $\order$ is also a valid input to polices in $\class_{\DTPNOI_{i^*}}$. Therefore by~\Cref{claim:DTPNOI} we know that 
    \begin{align*}
        U_{\DTPNOI_{i^*}}(\order) \geq U_{\TPNOI_{i^*}}(\order) - 2\eps \cdot \OPT \geq \OPT_{i^*} - 3\eps \cdot \OPT. 
    \end{align*}
\end{proof}

\claimLessThanNOPT*

\begin{proof}
    We know that one policy for the $\PNOI_{i^*}$ problem is to claim a box $i$ closed, pay no price, and get expected utility $\E[v_i]$. Since $\OPT$ is optimal among all possible policies for the $\PNOI_{i^*}$ problem, $OPT \geq \max_{i} \E[v_i]$. Since the values $v_i$ are at least $0$,
    \begin{align*}
        \E\sq*{\max_{i} v_i} \leq \E\sq*{\sum_{i} v_i} = \sum_{i} \E[v_i] \leq \sum_{i} \max_{i} \E[v_i] = n \cdot \max_{i} \E[v_i] \leq n \cdot \OPT. 
    \end{align*}
\end{proof}

\claimOPTSTDP*

\begin{proof}
    Let $t$ be the first iteration where $V_j>0$ and let $S^U$ be the set of boxes the agent ends up opening in rounds $> t$. The total reward the agent gets is just $\max \paren*{\max_{i \in S^U}\widetilde{\kappa_i}, V_j}$ (since $V_j > 0$, the final reward is $0$). In order to maximum this term, we should make the set $S^U$ as large as possible, namely, open all remaining boxes. 
\end{proof} 

\claimAssociatedALG*
 
\begin{proof}
    Let us first define the recurrence for the expected reward from $\ALG \in \class_{\STDP_{i^*}}$ parametrized by $\order$ for problem $\STDP_{i^*}$. Firstly, assuming that stage $j$ is the first step where $V_j > 0$, then $V_j = \widetilde{v_{i_j}}$, and we will open all unopened boxes. Let $\U_{j+1} := \M \setminus \{i_1, \cdots, i_j\}$. We know that then the reward $\ALG$ gets from steps $\geq (j+1)$ is just the maximum $\widetilde{(\kappa_{r} - V_j)^+}$ among $r \in \U_{j+1}$. This is equal to $(\widetilde{\weitz_{\geq (j+1)}} - \widetilde{v_{i_j}})^+$. Meanwhile, during step $j$, the reward $\ALG$ gains is simply $\widetilde{\kappa_{i_j}}$. Hence we can now get the following reward recurrence:
    \begin{align*}
        U_{\STDP_{i^*}}(\order)_{\geq j} = \Pr[v_{i_j} \leq \tau_j] \cdot U_{\STDP_{i^*}}(\order)_{\geq (j+1)} + \Pr[v_{i_j} > \tau_j] \cdot \paren*{\E_{v_{i_j} > \tau_j}[(\widetilde{\weitz_{\geq (j+1)}} - \widetilde{v_{i_j}})^+] + \E_{v_{i_j} > \tau_j}[\widetilde{\kappa_{i_j}}]}.
    \end{align*}
    Moreover, at the end of the policy, $U_{\STDP_{i^*}}(\order)_{\geq (k+1)} = \E[v_{i^*}]$. These recurrence specifications are exactly the same as for $\order$ from problem $\DTPNOI_{i^*}$. Therefore 
    \begin{align*}
        U_{\DTPNOI_{i^*}}(\order) = U_{\STDP_{i^*}}(\order). 
    \end{align*}
\end{proof}

\THMPTAS*
\begin{proof}%[Proof of \Cref{thm:ptas}]
Notice that given an adaptive algorithm $\ALG$ for problem $\STDP_{i^*}$, we could always find a corresponding policy in $\class_{\STDP_{i^*}}$ that has at least as much expected utility. Hence, given a PTAS to $\STDP_{i^*}$ problem (this is guaranteed to exist by~\cite{FuLX18}), we can get an $\order$ such that $$U_{\STDP_{i^*}}(\order) \geq (1 - \eps) \cdot U_{\STDP_{i^*}}(\order^*),$$ where $\order^*$ is the optimal \indexthreshold{} for problem $\STDP_{i^*}$.  
By~\Cref{claim:associatedALG} and ~\Cref{claim:DTPNOIopt}, we know that there exists a support $\W$ such that 
\begin{align*}
    U_{\STDP_{i^*}}(\order^*) = U_{\DTPNOI_{i^*}}(\order^*) \geq \OPT_{i^*} - 3 \eps \cdot \OPT. 
\end{align*}
Then
\begin{align*}
    U_{\STDP_{i^*}}(\order) \geq (1 - \eps) \cdot (\OPT_{i^*} - 3 \eps \cdot \OPT). 
\end{align*}
Now, we have already reasoned about the fact that given a fixed $\order$, our reformulations always had non-increasing expected utility compared to original formulation. Thus  $U_{\STDP_{i^*}}(\order) \leq U_{\PNOI_{i^*}}(\order)$. Thus for the $\order$ returned by our reduction from $\PNOI_{i^*}$, it must be the case that
\begin{align*}
    U_{\PNOI_{i^*}}(\order) \geq \OPT_{i^*} - O(\eps) \cdot \OPT. 
\end{align*}
Let us use $\ALG^{(i^*)}$ to denote the two-phase policy parametrized by $\order$ returned from problem $\PNOI_{i^*}$. Then we can conclude that by doing our reduction for $\PNOI_{i^*}$ for all $i^* \in \M$, then taking the better between the best $\ALG^{(i^*)}$ for all $i^* \in \M$ and Weitzman's policy, we can find a policy with reward at least $\OPT - O(\eps) \cdot \OPT$.

During our reduction to stochastic dynamic program, all steps are fully polynomial except from we tried all choices for $\W$, which takes $O\paren*{n^{\poly(1/\eps)}}$ time, which has polynomial dependence on $n$. Running the PTAS for the stochastic dynamic program itself also only takes time that has polynomial dependence on $n$. Therefore our policy finding scheme is a PTAS for the Pandora's box with nonobligatory inspection problem. 
\end{proof}

\end{document}